%% file: elicitation_arxiv.tex
\documentclass[11pt]{article}

\usepackage[T1]{fontenc}
\usepackage[margin=1in]{geometry}
\usepackage{amsmath,amssymb,amsthm,amsfonts,mathtools}
\usepackage{graphicx}
\usepackage{url}

\newtheorem{theorem}{Theorem}[section]
\newtheorem{corollary}{Corollary}[theorem]
\newtheorem{lemma}[theorem]{Lemma}

\newtheorem{prop}{Proposition}
\theoremstyle{remark}
\newtheorem{remark}[theorem]{Remark}

\newcommand{\E}{\mathrm{E}}

\newcommand{\Prob}{\mathrm{P}}

\newcommand{\eps}{\varepsilon}

\DeclarePairedDelimiter\paren{(}{)}

\DeclarePairedDelimiter\abs{\lvert}{\rvert}
\DeclarePairedDelimiter\norm{\lVert}{\rVert}
\DeclarePairedDelimiter\bkt{[}{]}
\DeclarePairedDelimiter\set{\{}{\}}

\title{A Least-Squares Approach to Sample-Based Prior Elicitation}
\author{Yannik Pitcan\thanks{The foundations of this work were developed while
the author was a Ph.D.\ candidate in the Department of Statistics, University
of California, Berkeley; the extensions in Sections~\ref{sec:scale},
\ref{sec:multivariate} and~\ref{sec:semisynthetic} were developed subsequently.
The author is an independent researcher. Correspondence: pitcany@gmail.com.}}
\date{\today}

\begin{document}
\maketitle

\begin{abstract}
An expert who supplies examples of a quantity often also signals how plausible
each one is; when is that signal worth using? We study eliciting
a Bayesian prior from an expert who provides example points together with their
approximate likelihoods. We propose fitting the prior by least squares---%
minimizing the squared discrepancy between a parametric density and the
elicited likelihoods---which defines an M-estimator that remains well posed even
for families whose moments do not exist. We establish
consistency and asymptotic normality,
and prove---under explicit regularity conditions, comprising a well-separation
and a uniform-concentration requirement that we verify for the families
considered---a non-asymptotic $O(1/\sqrt{n})$ Berry--Esseen bound on its
sampling distribution, uniform and nonuniform, by extending a result of
Pinelis for maximum-likelihood estimators to the M-estimation setting. We then
relax the assumptions that most limit the method in practice. Experts need not
report on the density's own scale: an unknown reporting scale can be profiled
out in closed form and estimated jointly, at no asymptotic cost for location
families. The theory extends to multivariate parameters, where a directional
Berry--Esseen bound follows from the multivariate delta method applied to a
smooth implicit proxy for the estimator. An additive error floor in the noise
model removes a degeneracy in the optimal design, making optimal designs
interior. Simulations for
normal and beta families confirm the predicted $n^{-1}$ error rate and the
accuracy of the normal approximation at moderate sample sizes. Finally, we
compare the estimator with the sample-only maximum-likelihood baseline,
derive an explicit threshold on the expert's reporting noise below which the
elicited likelihoods provably reduce estimation error, and calibrate that
threshold against eleven datasets of human frequency judgments.
\end{abstract}

\input{elicitation_body}

\section*{Code availability}

Code to reproduce every figure and every reported number is available at
\url{https://github.com/pitcany/prior-elicitation}, in the \texttt{experiments/}
directory: \texttt{ch5\_elicitation.py} generates
Figures~\ref{fig:elicitation-success}--\ref{fig:elicitation-normality},
\texttt{ch5\_ls\_vs\_mle.py} generates Figure~\ref{fig:ls-vs-mle},
\texttt{ch5\_design.py} produces the design comparison of
Section~\ref{sec:sample-based}, \texttt{ch5\_cauchy.py} produces the
no-moments results of Section~\ref{sec:cauchy}, \texttt{ch5\_scale.py}
produces Figure~\ref{fig:ch5-scale} and the sandwich checks of
Section~\ref{sec:scale}, \texttt{ch5\_multivariate.py} produces
Figure~\ref{fig:ch5-multivariate}, \texttt{ch5\_floor.py} produces the
error-floor designs of Remark~\ref{rem:floor}, and
\texttt{ch5\_semisynthetic.py} produces Figure~\ref{fig:ch5-semisynthetic}.
The human frequency-judgment data analysed in
Section~\ref{sec:semisynthetic} is included under
\texttt{experiments/data/risk\_judgments/} together with the script that
retrieves it from its source repository \cite{Pachur2024}. Package versions are
pinned in \texttt{requirements.txt}. Every script is seeded and reproduces the
reported figures and numbers exactly.

\end{document}

%% file: elicitation_body.tex
\section{Statistical Elicitation}

Elicitation is the process of forming a probability distribution from a person's knowledge and beliefs \cite{OHagan2006,Garthwaite2005}. 
We will focus on the case of elicitation to obtain a prior that will be used in a subsequent machine learning task. Although most of the results will be applicable to other motivations for elicitation, narrowing the language will simplify our discussion. 

Classically, elicitation is a human-centered process with multiple roles:
The \emph{modeler} will ultimately do the modeling, with the elicited prior. 
The \emph{facilitator} has a strategy and asks questions to gather information to use for inference.  
The \emph{expert} has the knowledge that the facilitator will use. 
A \emph{statistician} will train the expert on probability and provide feedback. 

An individual may fill multiple roles; for example, a single individual commonly fills both the statistician and facilitator roles.  
The expert may also be the modeler who will ultimately use the elicited prior.

Elicitation is a multi-stage process, typified by the following steps.  The modeler and the statistician will determine the target value in collaboration in the structuring and decomposition step. Then, during the elicitation phase there is further iteration over three steps: 1. elicit summaries, 2. fit a distribution, and 3. assess adequacy. 
The elicitation process is our focus for the presented work, as the fitting and assessment steps are the primary role of the automated tool. 

In higher dimensions, summaries are less intuitive and even cumbersome to communicate.  Therefore, we will, in our automated facilitator-statistician discussion, shift from eliciting summaries to eliciting samples.
Sample based elicitation has been applied in an experimental setting successfully for fitting distributions in commonly used univariate data models, going back to the comparative study of beta prior elicitation techniques by Winkler \cite{Winkler1967}; more recently, Casement and Kahle \cite{CasementKahle2018} elicit priors implicitly through an expert's selections among graphics of hypothetical future samples. See \cite{Mikkola2023} for a broad review of prior elicitation methods.

Literature on elicitation focuses on making inferences from the type of information provided by elicitation and the related psychological literature. The psychology of elicitation relates to how people characterize uncertainty (not consistently) to what information is actually needed in order to make inferences about uncertainty that are themselves useful for further inferences. 

In the human--computer interaction and visualization communities, elicitation has been studied with crowdworkers; for example, Goldstein and Rothschild \cite{GoldsteinRothschild2014} show that eliciting an entire distribution through a graphical frequency-based interface and computing statistics from it yields greater accuracy than asking for those statistics directly.

Toward the study of Bayesian modeling in a broad sense, HCI researchers have built tools for eliciting specific forms of priors \cite{CasementKahle2018,SarmaKay2020}.

Observing how statisticians set priors revealed that the choice of visualization can impact how experienced Bayesian statisticians choose to set a prior \cite{SarmaKay2020}. In designing a more general prior elicitation tool, it will be important to understand what forms of information will facilitate good inferences and to balance these forms with what psychological insights exist regarding how experts choose matching interfaces. Further research has examined what information experts are able to express reliably, and how visualizations impact the broad strategies of the expert. In this work, we consider the learnability of classes of priors from different forms of evidence toward making tool design choices. 


\section{Sample Based Elicitation}
\label{sec:sample-based}

Once elicited quantities are in hand, a distribution must be fitted to them, and practice has largely settled on least squares applied to the distribution function. The Sheffield framework chooses parameters by minimizing the squared discrepancy between elicited and fitted cumulative probabilities \cite{OakleyOHagan2019}, and fitting a parametric distribution to elicited summaries in this manner is standard across the literature \cite{Garthwaite2005,OHagan2006}. What is elicited, in each case, is a small number of \emph{summaries}: quantiles, probabilities, occasionally a mean.

Our proposal shares the least-squares principle but differs in both the elicited object and the residual. We elicit examples together with their reported plausibilities, and fit on the \emph{density} at those points rather than on the distribution function at elicited quantiles. The motivation is the one given above: summaries grow harder to communicate as the dimension increases, whereas examples remain natural to supply. Fitting on the density has a second consequence---the criterion stays well posed for families whose moments do not exist, and for which moment-based summaries are therefore unavailable---which we demonstrate in Section~\ref{sec:cauchy}.

This choice runs against a standing recommendation. Mikkola et al.\ \cite{Mikkola2023} observe that expressing knowledge in probabilistic terms is already hard for experts, ``let alone asking directly for the full density function,'' which is precisely why the field elicits summaries instead. We do not assume that experts report densities well. Their unreliability is the parameter $\sigma$ of the noise model $z=p_{\theta_0}(x)(1+\sigma\xi)$ used throughout Section~\ref{sec:elicitation-experiments}, and Section~\ref{sec:ls-vs-mle} quantifies how large $\sigma$ may be before the reports cease to be worth using---roughly $60\%$ relative error for the families we consider. The concern is therefore not set aside but priced.

A recent line of work automates the fitting step rather than changing what is elicited. Bockting et al.\ \cite{Bockting2024} learn the hyperparameters of a parametric prior by stochastic gradient descent on a discrepancy between model-implied and expert-elicited quantities, accommodating quantile-, moment- and histogram-based elicitation within a single framework, and extend the approach to non-parametric joint priors using normalizing flows \cite{Bockting2025}; Manderson and Goudie \cite{Manderson2023} pursue a related objective by multi-objective Bayesian optimization, eliciting in the space of observable quantities rather than that of model parameters. These methods share the premise of the present work---that fitting a prior is the minimization of a discrepancy between a parametric family and what the expert reports---and in one respect they are considerably more general than what we propose, since the elicited quantities need not be tied to the parameters being estimated. They differ in two respects that leave the contribution here complementary rather than competing. The elicited object is again a set of summaries, so the question of what to do with plausibilities attached to examples does not arise. And the guarantees are computational rather than statistical: the discrepancy is minimized numerically and the procedure validated on case studies; neither line proves asymptotic or finite-sample properties of the resulting estimate, and Manderson and Goudie note that their global-optimization step carries no guarantee of locating the optimum in finite time. Our residual is chosen so that the fit is an M-estimator, which is what buys the consistency and asymptotic normality of Section~\ref{sec:asymptotics} and the non-asymptotic bound of Section~\ref{sec:finite-sample}; and the comparison of Section~\ref{sec:ls-vs-mle} answers a question an automated fitting procedure does not raise, namely how noisy the expert may be before the elicited quantities cease to be worth incorporating at all.

\begin{remark}[Reports on an unknown scale]\label{rem:scale}
We take the reported $z_i$ to lie on the density's own scale, up to the multiplicative error $\sigma\xi$. An expert may instead report plausibilities on an arbitrary scale, $z=c\,p_{\theta_0}(x)(1+\sigma\xi)$ with $c>0$ unknown, in which case $\hat\theta$ as defined above is in general not consistent for $\theta_0$ and $c$ must be estimated jointly. Section~\ref{sec:scale} carries out the extension: the scale coordinate can be profiled out in closed form, so the joint fit costs nothing computationally, and estimating the scale costs nothing \emph{asymptotically} either for location families---for the beta shape family it even reduces the asymptotic variance slightly.
\end{remark}

In order to build a general automated elicitation tool, we need to consider how the tool will learn from the expert. In the end, this learning will be an online process which learns from each sample sequentially and then presents the updated model to the user for feedback. 
The estimator introduced below is a nonlinear least-squares fit in which the elicited points $x_i$ act as \emph{design points}, so their placement, and not merely their number, governs the precision of the fit: under a design measure $q$ the asymptotic variance is $\sigma^{2}\,\E_q[p_{\theta_0}^{2}\dot p_{\theta_0}^{2}]/\big(n(\E_q[\dot p_{\theta_0}^{2}])^{2}\big)$, and drawing the $x_i$ i.i.d.\ from the expert's belief is only one choice of $q$. In elicitation the examples are supplied by a human expert following instructions, and we expect them to be more spread out than an i.i.d.\ draw for two reasons. First, an expert is unlikely to give an example very close to one already given, which leads to a representative sample spanning the range of their belief. Second, the instructions can prompt the expert for examples that are both likely and unlikely. Spreading the design does help, in both families. At $n=30$ and $\sigma=0.1$, replacing the i.i.d.\ design of the normal location family by a uniform design on $[\theta_0-2,\theta_0+2]$ lowers the mean squared error by $32\%$, and an equispaced design on the same interval by $31\%$. For the beta shape family the designs must instead be subsets of $(0,1)$, since $\theta$ is a shape parameter; there a uniform design lowers the error by $17\%$ and an equispaced design by $13\%$. In each family the i.i.d.\ design is the least precise of those we tried.

Placement matters more than spread as such. A uniform design on $[\theta_0-3,\theta_0+3]$ in the normal family gains nothing over i.i.d., because the extra width places points where $\dot p_{\theta_0}\approx0$; and in the beta family a two-point design pairing an informative point at $x=0.05$ with a nearly uninformative one at $x=0.5$ is almost seven times \emph{worse} than i.i.d. Across the designs that improve on i.i.d., the measured error agrees with the asymptotic variance above to within $6\%$. We therefore read the i.i.d.\ experiments of Section~\ref{sec:elicitation-experiments} as a conservative reference point rather than a best case, while noting that we do not prove i.i.d.\ sampling is worst-case over all designs.

\begin{remark}[The optimal design is degenerate under this noise model]
\label{rem:degenerate-design}
The calculation above should not be read as advice on where to question an expert, because pursued to its conclusion it gives absurd advice. For a design concentrated at a single point $x_0$ the asymptotic variance reduces to $\sigma^{2}/\big(n\,s(x_0)^{2}\big)$, where $s:=\partial_{\theta}\log p_{\theta_0}$ is the score, so the best one-point design maximizes $\abs{s}$---and $\abs{s}$ is unbounded in both of our families: $s(x)=x-\theta_0$ for the normal location family, and $s(x)=1/\theta_0+1/(\theta_0+1)+\log x\to-\infty$ as $x\to0$ for the beta shape family. Concretely, the symmetric two-point design at $\theta_0\pm d$ in the normal family has asymptotic variance exactly $\sigma^{2}/(nd^{2})$, and we confirm this empirically out to $d=4$, where the error is sixteen times smaller than at $d=1$; correspondingly, one-point beta designs at $x_0=0.2$ down to $x_0=0.002$ have errors falling by a factor of $27$. The cause is the multiplicative form of the reporting model: the absolute error $\sigma p_{\theta_0}(x)$ vanishes wherever the density does, so a report made far out in the tail is treated as almost noiseless. A real expert asked for the plausibility of a value they consider impossible supplies no usable information at all. Remark~\ref{rem:floor} adds the missing ingredient---an additive error floor---and shows that it removes the degeneracy: with the floor in place the design problem has interior optima, and the design analysis above becomes usable advice rather than a cautionary tale.
\end{remark}

\begin{remark}[An additive floor removes the degeneracy]
\label{rem:floor}
Augment the reporting model with an additive error floor,
\[
z \;=\; p_{\theta_0}(x)\,(1+\sigma\xi)\;+\;\tau\eta,
\]
with $\tau>0$ and $\eta$ standard, independent of $x$ and $\xi$: a report about a value the expert considers implausible still carries error at least $\tau$. The conditional mean of $z$ is unchanged, so the estimator, its consistency and its asymptotic normality all carry over verbatim; only the conditional variance changes, to $\sigma^{2}p_{\theta_0}(x)^{2}+\tau^{2}$, and with it the design calculus. Under a design measure $q$ the asymptotic variance becomes
\[
V(q)\;=\;\frac{\E_q\big[(\sigma^{2}p_{\theta_0}^{2}+\tau^{2})\,\dot p_{\theta_0}^{2}\big]}{n\,\big(\E_q[\dot p_{\theta_0}^{2}]\big)^{2}},
\]
and under the i.i.d.-from-belief design it has the closed form $(\sigma^{2}A+\tau^{2}B)/(nB^{2})$ with $A$ and $B$ as in Section~\ref{sec:ls-vs-mle}: the floor enters as exactly $+\tau^{2}/(nB)$. Simulation confirms this within $3\%$ at $n=200$ for both families at $\tau\in\set{0.02,0.05}$, where the floor contributes between $2\%$ and $77\%$ of the total variance depending on the family and $\tau$.

The design problem now has interior solutions, because escaping to the tails sends $\dot p_{\theta_0}\to0$ while the numerator keeps its floor: a one-point design at $x_{0}$ has variance $(\sigma^{2}p_{\theta_0}(x_{0})^{2}+\tau^{2})/\big(n\,\dot p_{\theta_0}(x_{0})^{2}\big)\to\infty$ in the tails, instead of $\to0$. For the symmetric two-point design of the normal family, $V(d)=\big(\sigma^{2}\varphi(d)^{2}+\tau^{2}\big)/\big(n\,\varphi(d)^{2}d^{2}\big)$ is minimized at a finite $d^{\ast}$: at $\sigma=0.1$, $d^{\ast}=1.55$, $1.31$, $1.09$ for $\tau=0.01$, $0.02$, $0.05$---the optimum moves \emph{inward} as the floor grows, toward values the expert finds plausible. Empirical mean squared errors at $d^{\ast}$ match $V(d^{\ast})$ to within $7\%$ ($n=30$, $2{,}000$ replications), while the design at $d=4$ that was sixteen times better than $d=1$ without the floor is now roughly three orders of magnitude worse than $d^{\ast}$ (its predicted variance exceeds even what the bounded search interval allows the empirical error to express). The beta family behaves identically: the optimal one-point design sits at $x_{0}^{\ast}=0.115$, $0.152$, $0.210$ for the same $\tau$ values, and the tail design $x_{0}=0.002$ that was $27$ times better than $x_{0}=0.2$ without the floor is now about four orders of magnitude worse than the optimum. Finally, the comparison of Section~\ref{sec:ls-vs-mle} extends unchanged in form: the least-squares variance under the i.i.d.\ design is $(\sigma^{2}A+\tau^{2}B)/(nB^{2})$, so the region in which reported plausibilities beat the sample-only MLE is the ellipse $\set{(\sigma,\tau):\ \sigma^{2}A/B^{2}+\tau^{2}/B<1/I_F(\theta_0)}$, of which the crossover $\sigma^{\ast}$ is the $\tau=0$ section.
\end{remark}

In this section, we present our main analytical results. 
First, we will introduce our least squares based objective function. 
Next, we will consider the large sample behavior of the proposed estimator 
by evaluating the consistency of the estimator, and we will show the conditions under which we achieve asymptotic normality. Third, we present a finite sample result. 

\section{A Least-Squares Based Approach to Elicitation}

\subsection{Proposed Objective Function}\label{sec:objective}

Assume that we elicit i.i.d. observations $x_i$ with corresponding sample likelihoods $z_i$ for $i=1,\ldots,n$. Assuming we have a parametric model class, our proposed method of estimating $\theta$ involves minimizing an objective function, which we illustrate below.

Let $Q((\vec{x},\vec{z}),\theta)=\sum_i \left( l(x_i,\theta)-z_i \right)^2$

Our proposed optimization problem is $$
\hat{\theta} = \arg \min_{\theta} \sum_i \left( l(x_i,\theta)-z_i \right)^2 = \arg \min_{\theta} Q((\vec{x},\vec{z}), \theta),$$

where $x_i,z_i$ is the $i$th sample and likelihood.

$$\frac{dQ}{d\theta} = 2 \sum_i \left( (l(x_i,\theta)-z_i) \frac{\partial}{\partial \theta}l(x_i,\theta)\right) $$

and let $\psi((x_i,z_i),\theta)=\big(l(x_i,\theta)-z_i\big)\frac{\partial}{\partial \theta}l(x_i,\theta),$ so that $\frac{dQ}{d\theta}=2\sum_i \psi((x_i,z_i),\theta)$ and $\hat\theta$ solves $\sum_i \psi((x_i,z_i),\theta)=0$: the estimating function, not the raw residual.

$\hat{\theta}$ is a solution to $$\sum_i \left( (l(x_i,\theta)-z_i) \frac{\partial}{\partial \theta}l(x_i,\theta)\right) = 0$$

and

$\theta_0$ solves $$\E_{\theta_0}\left[(l(x_i,\theta)-z_i) \frac{\partial}{\partial \theta}l(x_i,\theta)\right]=0$$

\section{Asymptotic Analysis}
\label{sec:asymptotics}

Let $\Omega$ be the parameter space with an open set $\omega$ such that $\theta_0$, the true parameter value, is an interior point.

\subsection{Consistency}
\label{sec:consistency}

We obtain consistency from the standard pair of conditions for M-estimation: that $\theta_{0}$ be well separated from the rest of $\Theta$ under the population criterion, and that the sample criterion converge to it uniformly. Recall $L_n(\theta)=\sum_{i=1}^{n}\ell_{X_i,Z_i}(\theta)$ with $\ell_{x,z}(\theta)=-(l_x(\theta)-z)^2$.

\begin{enumerate}
	\item[(A5)] \emph{(Well-separation.)} For each $\delta>0$,
	$$
	D(\delta) := \E\,\ell_{X,Z}(\theta_0)-\!\!\sup_{\theta\in\Theta:\,\abs{\theta-\theta_0}\ge\delta}\!\!\E\,\ell_{X,Z}(\theta)\ >\ 0 .
	$$
	\item[(A6)] \emph{(Uniform concentration.)} There exist $C_1,c_2\in(0,\infty)$, not depending on $n$, such that
	$$
	\Prob\Big(\sup_{\theta\in\Theta}\abs*{n^{-1}L_n(\theta)-\E\,\ell_{X,Z}(\theta)}\ \ge\ D(\delta)/2\Big)\ \le\ C_1 e^{-c_2 n}.
	$$
\end{enumerate}

Under (A5) and (A6) the estimator is consistent, $\hat{\theta} \xrightarrow{\mathcal{P}}\ \theta_0$; the argument given in Section~\ref{sec:remainder} in fact yields the stronger conclusion $\Prob(\abs{\hat\theta-\theta_0}>\delta)\le C_1e^{-c_2 n}$ for each $\delta>0$. Condition (A5) has a closed form and requires no compactness assumption, and (A6) is proved for the families used here in Section~\ref{sec:appendix-a6}; both are discussed in Remark~\ref{rem:concavity}. Compactness is not essential to (A6) either: Proposition~\ref{prop:a6-unbounded} establishes it on $\Theta=\mathbb{R}$ for the normal location family, and where (A6) genuinely fails on an unbounded set, as it does for the beta shape family on $[1.1,\infty)$, Proposition~\ref{prop:consistency-unbounded} recovers the exponential consistency bound by a one-sided argument.

One might instead hope to argue from monotonicity of $\theta\mapsto(l(x,\theta)-z)\,\partial_{\theta}l(x,\theta)$ together with continuity near $\theta_0$ and an isolated root there, as is common for estimating equations. That route is unavailable here: this function tends to $0$ as $\abs{\theta}\to\infty$ whenever $p_{\theta}(x)$ and $\partial_{\theta}p_{\theta}(x)$ do, and a monotone function with equal limits at $\pm\infty$ is constant, so it is monotone for no such family. The same obstruction rules out the concavity hypothesis used by \cite{Pinelis2017}, as discussed in Remark~\ref{rem:concavity}.

\subsection{Asymptotic Normality}

$$\frac{\partial}{\partial \theta} \psi((x,z),\theta) = \left[ \frac{\partial}{\partial \theta}l(x,\theta) \right]^2  + (l(x,\theta)-z) \frac{\partial^2}{\partial \theta^2}l(x,\theta)$$

Assume the observations $(X_i,Z_i)$ are i.i.d., $\theta_0$ is an interior point, and the consistent estimator satisfies the estimating equation with probability tending to one. In addition, assume
\[
\E\psi((X,Z),\theta_0)=0,\qquad
\E\psi((X,Z),\theta_0)^2<\infty,\qquad
0<\abs{\E\partial_\theta\psi((X,Z),\theta_0)}<\infty.
\]
On a neighborhood $N$ of $\theta_0$, require $\theta\mapsto\psi((X,Z),\theta)$ to be continuously differentiable almost surely, with an integrable envelope for its derivative:
\[
\E\sup_{\theta\in N}\abs{\partial_\theta\psi((X,Z),\theta)}<\infty.
\]
These local conditions give the uniform law of large numbers needed for the derivative in the estimating-equation expansion. Then $$\sqrt{n}(\hat{\theta}-\theta_0) \xrightarrow{\mathcal{L}}\ \mathcal{N}(0,\sigma_{\hat{\theta}}^2)$$

where $$\sigma_{\hat{\theta}}^2 = \frac{\E_{\theta_0}[\psi^2((X,Z),\theta_0)]}{(\E_{\theta_0} [\frac{\partial}{\partial \theta} \psi((X,Z),\theta)|_{\theta=\theta_0}])^2}.$$

\subsection{When do the reported likelihoods help? A comparison with maximum likelihood}
\label{sec:ls-vs-mle}

The estimator above uses both the examples $x_i$ and the reported likelihoods $z_i$. A natural question is what the $z_i$ buy us relative to the obvious sample-only alternative: maximum likelihood on the examples alone. When the examples are drawn i.i.d.\ from the expert's belief $p_{\theta_0}$, the MLE $\hat{\theta}_{\mathrm{MLE}}=\arg\max_\theta \sum_i \log p_\theta(x_i)$ is Cram\'er--Rao efficient \emph{among all estimators that use only the samples}, with asymptotic variance $1/(nI_F(\theta_0))$, where $I_F(\theta_0)=\E_{\theta_0}[(\partial_\theta \log p_{\theta_0})^2]$ is the Fisher information. Comparing our estimator against this baseline therefore isolates the value of the reported likelihoods.

We adopt the noise model of Section \ref{sec:elicitation-experiments}, $z=p_{\theta_0}(x)(1+\sigma\xi)$ with $\E[\xi]=0$, $\E[\xi^2]=1$, and $\xi$ independent of $x$, where $\sigma$ measures the expert's unreliability. Writing $\dot p_\theta = \partial_\theta p_\theta$ and specializing $\sigma_{\hat\theta}^2$ to this model gives, at $\theta_0$,
\[
\psi((x,z),\theta_0) = \big(p_{\theta_0}(x)-z\big)\dot p_{\theta_0}(x) = -\sigma\xi\, p_{\theta_0}(x)\,\dot p_{\theta_0}(x),
\]
so that, using independence of $\xi$ and $x$,
\[
I_1(\theta_0)=\E[\psi^2]=\sigma^2 A,\quad A:=\E_{\theta_0}\!\big[p_{\theta_0}^2\,\dot p_{\theta_0}^2\big], \qquad
I_2(\theta_0)=\E[\partial_\theta\psi]=B,\quad B:=\E_{\theta_0}\!\big[\dot p_{\theta_0}^2\big],
\]
where the cross term in $I_2$ vanishes because $\E[(p_{\theta_0}-z)\mid x]=0$.

\begin{remark}[Normalization of $I_1$ and $I_2$]
Throughout we define $I_1(\theta_0):=\E[\psi^2((X,Z),\theta_0)]$ and
$I_2(\theta_0):=\E[\partial_\theta\psi((X,Z),\theta_0)]$ in terms of the
estimating function $\psi$. Section~\ref{sec:setting} states the same two
quantities in terms of the per-observation criterion
$\ell_{x,z}(\theta)=-(l_x(\theta)-z)^2$, for which
$\ell_{x,z}'=-2\psi$ and $\ell_{x,z}''=-2\partial_\theta\psi$; that convention
therefore yields $4I_1$ and $2I_2$. Only the ratio $I_2(\theta_0)^2/I_1(\theta_0)$
enters the standardization and the asymptotic variance, and it is identical under
both conventions, so no result depends on the choice.
\end{remark}

Hence the least-squares estimator has asymptotic variance
\[
\mathrm{Var}_{\mathrm{LS}}(\hat\theta)\ \sim\ \frac{\sigma^2 A}{nB^2},
\qquad\text{compared with}\qquad
\mathrm{Var}_{\mathrm{MLE}}(\hat\theta)\ \sim\ \frac{1}{nI_F(\theta_0)}.
\]
Two features stand out. First, $\mathrm{Var}_{\mathrm{LS}}\propto\sigma^2$: as the expert becomes reliable ($\sigma\to0$) the variance vanishes, consistent with noise-free exact recovery when the realized density evaluations identify the parameter, whereas the MLE variance is a fixed constant no matter how good the expert is---the samples alone cannot pin down a density they were merely drawn from beyond the parametric rate. Second, the two estimators cross at
\[
\boxed{\ \sigma^\ast = \frac{B}{\sqrt{A\,I_F(\theta_0)}}\ }
\]
Below $\sigma^\ast$ the reported likelihoods strictly improve on the best sample-only estimator; above it, a reliable expert's samples are worth more than noisy plausibility reports and one should fall back on maximum likelihood.

Figure \ref{fig:ls-vs-mle} evaluates this prediction. For the normal location family $\sigma^\ast\approx0.64$ and for the beta shape family $\sigma^\ast\approx0.66$: in both cases the least-squares estimator dominates maximum likelihood for expert noise up to roughly $60\%$ relative error in the reported likelihoods, a regime that comfortably covers a competent expert. Below $\sigma^\ast$ the empirical mean squared errors track the asymptotic curves to within $21\%$ (normal) and $10\%$ (beta); above $\sigma^\ast$ the least-squares curve rises faster than its asymptotic approximation at this sample size. The empirical crossover falls at $0.58$ for the normal family and $0.65$ for the beta, against predicted values of $0.64$ and $0.66$. We emphasize that this comparison is deliberately \emph{conservative} for our method: the i.i.d.-from-belief design is precisely where the sample-only baseline is strongest, because it makes the example locations themselves an efficient encoding of the density. In the diverse-sampling regime for which sample-based elicitation is actually motivated---where an expert deliberately supplies spread-out, representative examples rather than i.i.d.\ draws---the example locations no longer encode $p_{\theta_0}$, maximum likelihood on them is inconsistent for the belief, and the reported likelihoods become indispensable rather than merely helpful. The crossover $\sigma^\ast$ should therefore be read as a lower bound on the range of expert noise for which the proposed method is preferable.

\begin{figure}
	\centering
	\includegraphics[width=\linewidth]{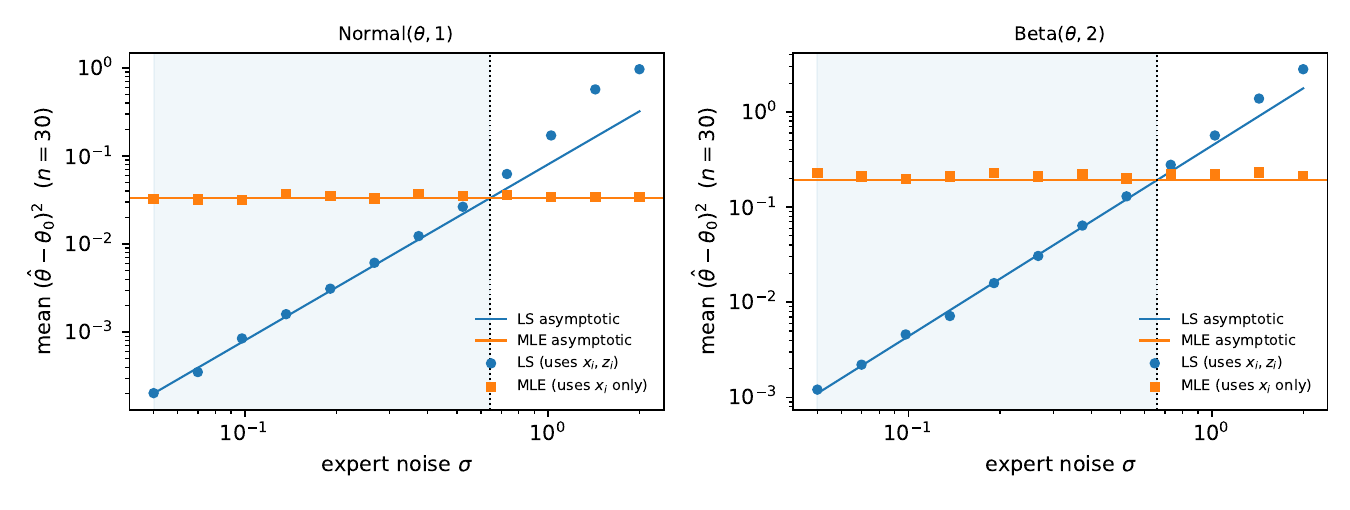}
	\caption[Least-squares elicitation versus maximum likelihood.]{Mean squared error of the least-squares elicitation estimator (using $x_i$ and the reported likelihoods $z_i$) and of the sample-only maximum-likelihood estimator (using $x_i$ alone), as a function of expert noise $\sigma$, at $n=30$ (log--log; markers are empirical MSE over $800$ replications, lines are the asymptotic variances $\sigma^2A/(nB^2)$ and $1/(nI_F)$). The two cross at the predicted $\sigma^\ast$ (dotted); in the shaded region the reported likelihoods strictly help. Left: $\mathcal{N}(\theta,1)$. Right: $\mathrm{Beta}(\theta,2)$.}
	\label{fig:ls-vs-mle}
\end{figure}

\subsection{Estimating the reporting scale}
\label{sec:scale}

The comparison above, like the rest of the chapter so far, takes the reported plausibilities to lie on the density's own scale. Following Remark~\ref{rem:scale}, we now drop that assumption: the expert reports
\[
z \;=\; c_{0}\,p_{\theta_{0}}(x)\,(1+\sigma\xi),
\]
with the reporting scale $c_{0}>0$ unknown, so that only \emph{relative} plausibility is assumed meaningful. This is the reporting format the elicitation literature considers realistic---the objection of \cite{Mikkola2023} to density elicitation is precisely that absolute density values are not available to introspection---so the results of this section remove the largest idealization in the noise model.

The fixed-scale estimator of Section~\ref{sec:objective} fails under this model, and not gracefully. Its population criterion becomes $\theta\mapsto\E\,\big(p_{\theta}(X)-c_{0}p_{\theta_{0}}(X)\big)^{2}$ up to an additive constant, and the minimizer generally sits away from $\theta_{0}$. For the beta shape family of Section~\ref{sec:elicitation-experiments} at $c_{0}=1.5$ the population minimizer is $3.495$ against $\theta_{0}=3$---a bias of twice the success tolerance used there, so the fit \emph{always} fails asymptotically---and at $c_{0}=0.5$ the minimizer is pinned to the boundary of $\Theta=[1.1,10]$. For the normal location family, writing $d=\theta-\theta_0$ and $\kappa=1/(2\pi\sqrt{3})$, the population criterion is
\[
R(\theta)=\kappa(1-2c_0)e^{-d^2/3}+\text{constant}.
\]
Thus $c_0>1/2$ gives the unique minimizer $\theta_0$, $c_0=1/2$ makes the population criterion completely flat, and $c_0<1/2$ puts the minimum at the point or points farthest from $\theta_0$ on a compact search interval. On $\Theta=\mathbb{R}$ in the last case, the infimum is approached as $\abs{\theta-\theta_0}\to\infty$ and is not attained at a finite parameter. In particular, at $c_0=0.5$ there is no uniquely shifted population minimizer. Empirically, at $c_{0}=1.5$ the mean squared error of the fixed-scale fit for the beta family plateaus at $0.248$ by $n=320$---the squared population bias---while the joint estimator below continues to decay at the $n^{-1}$ rate (fitted log--log slope $-1.11$ in both families; Figure~\ref{fig:ch5-scale}).

The remedy is to estimate the scale jointly,
\[
(\hat\theta,\hat c)\;=\;\arg\min_{\theta\in\Theta,\ c\in\mathcal{C}}\ \sum_{i=1}^{n}\big(c\,p_{\theta}(x_i)-z_i\big)^{2},
\]
with $\mathcal{C}=[c_{\mathrm{lo}},c_{\mathrm{hi}}]\subset(0,\infty)$ compact and $c_0\in\operatorname{int}\mathcal{C}$ for the asymptotic normality result below. The extension costs nothing computationally: for fixed $\theta$ the objective is linear least squares in $c$, so
\[
\hat c(\theta)\;=\;\Pi_{\mathcal{C}}\!\left(\frac{\sum_i z_i\,p_{\theta}(x_i)}{\sum_i p_{\theta}(x_i)^{2}}\right),
\]
where $\Pi_{\mathcal{C}}(u)=\min\{c_{\mathrm{hi}},\max\{c_{\mathrm{lo}},u\}\}$ is projection onto $\mathcal{C}$, provided the denominator is positive. If it is zero, the objective is constant in $c$ and any $c\in\mathcal{C}$ minimizes it. The joint fit is the same one-dimensional grid-plus-refinement search as before, applied to the profiled objective. Identifiability is inherited from the family, under one support condition: assume every $p_{\theta}$, $\theta\in\Theta$, has the same support as $p_{\theta_{0}}$ (true of both families used here---all of $\mathbb{R}$ for the normal location family, $(0,1)$ for the beta shape family). The population criterion is $\E\big(c\,p_{\theta}-c_{0}p_{\theta_{0}}\big)^{2}$ plus a constant; it vanishes only if $c\,p_{\theta}=c_{0}p_{\theta_{0}}$ almost everywhere on that common support, and integrating both sides over it---each density integrating to one---forces $c=c_{0}$, hence $\theta=\theta_{0}$ by identifiability of the family. For well-separation, complete the square in $c$ at fixed $\theta$:
\[
\E\big(c\,p_{\theta}-c_{0}p_{\theta_{0}}\big)^{2}
\;=\;\E[p_{\theta}^{2}]\,\big(c-c^{*}(\theta)\big)^{2}
\;+\;c_{0}^{2}\,\E[p_{\theta_{0}}^{2}]\,\big(1-\rho(\theta)^{2}\big),
\qquad
\rho(\theta)\;=\;\frac{\E[p_{\theta}\,p_{\theta_{0}}]}{\sqrt{\E[p_{\theta}^{2}]\;\E[p_{\theta_{0}}^{2}]}},
\]
with $c^{*}(\theta)=c_{0}\,\E[p_{\theta}p_{\theta_{0}}]/\E[p_{\theta}^{2}]$. Condition (A5) for the pair then follows from one ingredient per term. Pairs whose $\theta$-coordinate is $\delta$-far from $\theta_{0}$ are separated by the second term whenever $\sup_{\abs{\theta-\theta_{0}}\ge\delta}\rho(\theta)<1$. This sufficient condition follows from continuity and tail separation for the families considered here, but pointwise identifiability alone does not imply it on arbitrary parameter spaces: correlations can approach one along a sequence of parameters bounded away from $\theta_0$, without attaining one. Pairs whose $\theta$-coordinate is close to $\theta_{0}$ but whose scale is not are separated by the first term: $\theta\mapsto\E[p_{\theta}^{2}]$ and $\theta\mapsto c^{*}(\theta)$ are continuous at $\theta_{0}$ with $\E[p_{\theta_{0}}^{2}]>0$ and $c^{*}(\theta_{0})=c_{0}$, so on a small enough $\theta$-neighborhood the first term is bounded below by a positive multiple of $(c-c_{0})^{2}$. Both ingredients hold for the families used here. The uniform-concentration condition (A6) extends to the rectangle $\Theta\times\mathcal{C}$ with the same exponential rate; the covering argument is unchanged except that the net has $N_{1}N_{2}$ points (Remark~\ref{rem:a6-scale}), and consistency of $(\hat\theta,\hat c)$ follows exactly as in Section~\ref{sec:consistency}.

For the asymptotic distribution, write $\beta=(\theta,c)$, $m_{\beta}(x)=c\,p_{\theta}(x)$ and $\nabla m_{\beta}=(c\,\dot p_{\theta},\,p_{\theta})^{\top}$; the estimating equation is $\sum_i\big(m_{\beta}(x_i)-z_i\big)\nabla m_{\beta}(x_i)=0$, and since $\E[z\mid x]=c_{0}p_{\theta_{0}}(x)$, the usual sandwich argument gives, under the two-parameter analogues of the smoothness and moment conditions of Section~\ref{sec:setting},
\[
\sqrt{n}\,\big(\hat\beta-\beta_{0}\big)\ \xrightarrow{\mathcal{L}}\ \mathcal{N}\big(0,\ \sigma^{2}c_{0}^{2}\,G^{-1}WG^{-1}\big),
\qquad
G:=\E\big[\nabla m\,\nabla m^{\top}\big],\quad
W:=\E\big[p_{\theta_{0}}^{2}\,\nabla m\,\nabla m^{\top}\big],
\]
with $\nabla m$ evaluated at $\beta_{0}$ and all expectations under $X\sim p_{\theta_{0}}$; the extra factor $p_{\theta_{0}}^{2}$ in $W$ is the heteroscedasticity of the reports, $\mathrm{Var}(z\mid x)=\sigma^{2}c_{0}^{2}\,p_{\theta_{0}}(x)^{2}$. Writing $a:=\E[p_{\theta_{0}}^{3}\dot p_{\theta_{0}}]$, $b:=\E[p_{\theta_{0}}\dot p_{\theta_{0}}]$, $m_{2}:=\E[p_{\theta_{0}}^{2}]$, $m_{4}:=\E[p_{\theta_{0}}^{4}]$ and $A$, $B$ as in Section~\ref{sec:ls-vs-mle}, the $\theta$-coordinate of the sandwich is, explicitly,
\begin{equation}\label{eq:vtheta-scale}
V_{\theta}\;=\;\sigma^{2}\,\frac{A\,m_{2}^{2}-2ab\,m_{2}+b^{2}m_{4}}{\big(B\,m_{2}-b^{2}\big)^{2}} .
\end{equation}
Three consequences of \eqref{eq:vtheta-scale} deserve notice. First, $V_{\theta}$ does not depend on $c_{0}$: the asymptotic precision of $\hat\theta$ is invariant to the units of the reports when the true scale is interior to $\mathcal{C}$. Exact finite-sample equivariance under rescaling the reports also requires rescaling $\mathcal{C}$, since otherwise its boundary can become active. Second, for \emph{any} location family whose density vanishes in the tails, $b=-\int f^{2}f'=0$ and $a=-\int f^{4}f'=0$, so \eqref{eq:vtheta-scale} collapses to $V_{\theta}=\sigma^{2}A/B^{2}$: the known-scale asymptotic variance of Section~\ref{sec:ls-vs-mle}, exactly. Estimating the reporting scale is asymptotically free for location families. Third, the ratio $R:=V_{\theta}B^{2}/(\sigma^{2}A)$ need not exceed one: for $\mathrm{Beta}(\theta,2)$ at $\theta_{0}=3$ it is $R=0.958$, so the scale-free estimator is asymptotically \emph{more} precise than the fixed-scale fit it replaces. There is no contradiction---least squares is not efficient under the heteroscedastic noise above, and the reports are noisiest in the direction of $p_{\theta_{0}}$ itself; the scale coordinate absorbs part of that component of the noise instead of letting it contaminate $\hat\theta$. Repeating the comparison of Section~\ref{sec:ls-vs-mle} with $V_{\theta}$ in place of $\sigma^{2}A/B^{2}$ moves the crossover for the beta family from $\sigma^{\ast}=0.660$ to $0.674$ and leaves the normal family's $0.643$ unchanged: dropping the absolute-scale assumption does not shrink the regime in which reported plausibilities help.

Simulation confirms the sandwich. At $n=200$ over $2{,}000$ replications with $\sigma=0.1$ and $c_{0}\in\set{0.5,1,2}$, the empirical values of $n\,\mathrm{Var}(\hat\theta)$ and $n\,\mathrm{Var}(\hat c)$ agree with the corresponding diagonal entries of $\sigma^{2}c_{0}^{2}G^{-1}WG^{-1}$ to within $7\%$ in every configuration and both families, the empirical $\theta$-variances are indistinguishable across the three values of $c_{0}$ as \eqref{eq:vtheta-scale} requires, and the errors standardized by the sandwich have mean at most $0.05$ and standard deviation within $0.04$ of $1$. The beta-family ratio $R$ is confirmed directly: over $8{,}000$ replications at $c_{0}=1$ the joint estimator's variance is $0.962$ times the fixed-scale estimator's, against the asymptotic $0.958$.

\begin{figure}
	\centering
	\includegraphics[width=\linewidth]{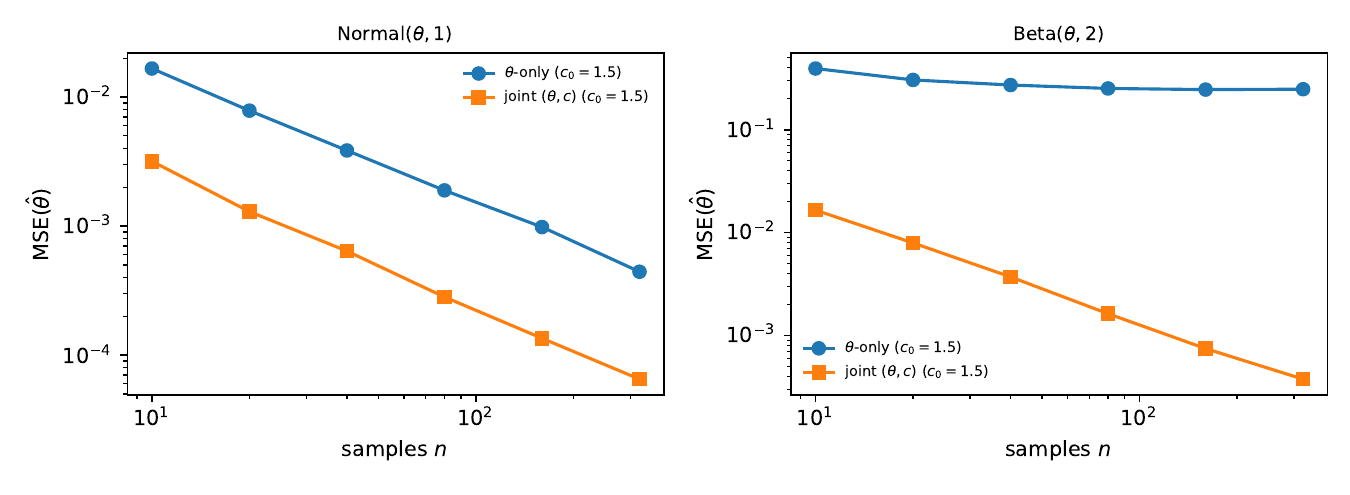}
	\caption[Joint estimation of the parameter and the reporting scale.]{Mean squared error of $\hat\theta$ versus samples per batch when the expert reports on a misspecified scale ($c_{0}=1.5$, $\sigma=0.1$, log--log, $400$ replications per point). The fixed-scale estimator (circles) plateaus at its squared population bias; the joint estimator of Section~\ref{sec:scale} (squares) continues at the $n^{-1}$ rate. Left: $\mathcal{N}(\theta,1)$, where the plateau is invisible at this $c_{0}$ because symmetry keeps the fixed-scale fit consistent for $c_{0}>1/2$. Right: $\mathrm{Beta}(\theta,2)$, where the plateau equals the squared population bias $0.495^{2}\approx0.245$.}
	\label{fig:ch5-scale}
\end{figure}

\subsection{Finite Sample}
\label{sec:finite-sample}

To do an elicitation, we will need to obtain a finite number of samples from an expert. While the large sample results give confidence in the general tractability of the problem, the finite sample results are important to understanding the realistic feasibility of implementing an automated elicitation tool.  

For the finite sample result, the assumptions are the following \cite{Pinelis2017}:

Let $\theta_{0}\in \Theta$ be the expert's target value of the parameter $\theta$, such that
\begin{center}
	$[\theta_{0}-\delta,\ \theta_{0}+\delta]\subseteq \theta^{\circ}$ 
\end{center}
for some real $\delta>0$, where $\theta^{\circ}$ denotes the interior of the subset $\Theta$ of $\mathbb{R}$. 

For brevity, we will use $\mathrm{P}$ and $\mathrm{E}$ throughout defined as $\mathrm{P} :=\mathrm{P}_{\theta_{0}}$ and $\mathrm{E} :=\mathrm{E}_{\theta_{0}}.$

For $x\in \mathcal{X}$, $z\in\mathcal{Z}$ and $\theta\in \Theta$, consider the
\emph{per-observation} criterion
$$
\ell_{x,z}(\theta)=- (l_x(\theta)-z)^2 ,
$$
and write $L_n(\theta):=\sum_{i=1}^{n}\ell_{X_i,Z_i}(\theta)$ for the sample
criterion that $\hat\theta$ maximizes. The assumptions below constrain the
per-observation criterion $\ell_{X,Z}$, so that $I_1(\theta_0)$ and
$I_2(\theta_0)$ are fixed $O(1)$ constants; cf.\ the remark on the
normalization of $I_1$ and $I_2$ in Section~\ref{sec:ls-vs-mle}.

\begin{enumerate}
	\item The set $\mathcal{X}_{>0} :=\{x\in \mathcal{X}:p_{\theta}(x)>0\}$ is the same for all $\theta\in [\theta_{0}-\delta,\ \theta_{0}+\delta],$ and for each $x\in \mathcal{X}_{>0}$ the likelihood $l_{x}(\theta)$ are thrice differentiable in $\theta$ at each point $\theta\in [\theta_{0}-\delta,\ \theta_{0}+\delta].$
	\item $\mathrm{E}\ell_{X,Z}'(\theta_0)=0$, $0<\mathrm{E}[\ell_{X,Z}'(\theta_0)^2]=I_1(\theta_0)<\infty$, and $0<-\mathrm{E}\ell_{X,Z}''(\theta_0)=I_2(\theta_0)<\infty$.
	\item $\mathrm{E}|\ell_{X,Z}'(\theta_{0})|^{3}+\mathrm{E}|\ell_{X,Z}^{''}(\theta_{0})|^{3}<\infty.$
	\item $\mathrm{E} \sup |\ell_{X,Z}^{'''}(\theta)|^{3}<\infty.$
	$$
	\theta\in [\theta_0-\delta,\theta_0+\delta]
	$$
\end{enumerate}

Suppose that the above conditions hold, together with the well-separation and uniform-concen\-tration conditions (A5) and (A6) of Section~\ref{sec:consistency}.

Then

$$|\mathrm{P}\left(\sqrt{n \frac{I_2(\theta_{0})^2}{I_1(\theta_0)}}(\hat{\theta}-\theta_{0})\leq z\right)-\Phi(z)|\leq\frac{C}{\sqrt{n}}$$

for all real $z$, and

$$|\mathrm{P}(\sqrt{n\frac{I_2(\theta_{0})^2}{I_1(\theta_0)}}(\hat{\theta}-\theta_{0})\leq z)-\Phi(z)|\leq\frac{C_{\omega}}{z^3 \sqrt{n}}$$

for $z\in (0,\omega \sqrt{n}]$ for any $\omega\in (0,\infty)$. $C_{\omega}$ is a finite expression that depends on $\omega$ and neither $C$ and $C_{\omega}$ depend on $n$ or $z$.

\subsection{Multivariate parameters}
\label{sec:multivariate}

The restriction to scalar $\theta$ is expository, and the elicitation problem one actually faces is multivariate: a prior has at least a location and a scale. Let now $\Theta\subseteq\mathbb{R}^{k}$ with $\theta_{0}$ an interior point, and let $\hat\theta$ minimize the least-squares criterion over $\Theta$. To cover the reporting-scale estimator of Section~\ref{sec:scale} at the same time we state the results for a general smooth mean function $m_{\theta}(x)$ with $\E[z\mid x]=m_{\theta_{0}}(x)$ and $\mathrm{Var}(z\mid x)=\sigma^{2}m_{\theta_{0}}(x)^{2}$; the elicitation estimator is the case $m_{\theta}=p_{\theta}$, and the pair $(\theta,c)$ of Section~\ref{sec:scale} is the case $m_{(\theta,c)}=c\,p_{\theta}$ with parameter dimension $k+1$. The estimating function is
\[
\psi\big((x,z),\theta\big)\;=\;\big(m_{\theta}(x)-z\big)\,\nabla m_{\theta}(x)\ \in\ \mathbb{R}^{k},
\]
and $\hat\theta$ solves $\sum_{i}\psi((x_{i},z_{i}),\theta)=0$. We assume the multivariate analogues of the conditions of Section~\ref{sec:setting}:
\begin{enumerate}
	\item[(M1)] $m_{\theta}(x)$ is thrice continuously differentiable in $\theta$ on a ball $\bar B(\theta_{0},\delta)\subseteq\Theta^{\circ}$, for each $x$;
	\item[(M2)] $I_{2}:=\E\big[\nabla m_{\theta_{0}}\nabla m_{\theta_{0}}^{\top}\big]$ is nonsingular, and $I_{1}:=\E\big[\psi\psi^{\top}\big]=\sigma^{2}\,\E\big[m_{\theta_{0}}^{2}\,\nabla m_{\theta_{0}}\nabla m_{\theta_{0}}^{\top}\big]$ is finite;
	\item[(M3)] $\E\norm{\psi(\theta_{0})}^{3}+\E\norm{\nabla_{\theta}\psi(\theta_{0})}^{3}<\infty$;
	\item[(M4)] $\E\sup_{\theta\in\bar B(\theta_{0},\delta)}\norm{\nabla^{2}_{\theta}\psi(\theta)}^{3}<\infty$,
\end{enumerate}
together with (A5) and (A6), which are stated in terms of $\abs{\theta-\theta_{0}}$ and a supremum over $\Theta$ and hence make sense verbatim with the Euclidean norm; on a bounded box the covering proof of Section~\ref{sec:appendix-a6} goes through with a $\prod_{j}N_{j}$-point net (Remark~\ref{rem:a6-scale}), and the argument of Section~\ref{sec:remainder} again yields exponential consistency, $\Prob(\abs{\hat\theta-\theta_{0}}>\delta)\le C_{1}e^{-c_{2}n}$. In (M2) the Hessian cross term $\E[(m_{\theta_{0}}-z)\nabla^{2}m_{\theta_{0}}]$ vanishes because $\E[z\mid x]=m_{\theta_{0}}(x)$, which is why $I_{2}$ is the outer-product matrix rather than a difference of two terms.

Under these conditions the classical sandwich argument gives
\[
\sqrt{n}\,\big(\hat\theta-\theta_{0}\big)\ \xrightarrow{\mathcal{L}}\ \mathcal{N}\big(0,\ V\big),
\qquad V=I_{2}^{-1}I_{1}I_{2}^{-1},
\]
of which the $2\times2$ display of Section~\ref{sec:scale} is an instance. The finite-sample question is the multivariate analogue of the Berry--Esseen bound above, and here the univariate proof does not transfer: the bracketing of Section~\ref{sec:remainder} solves the quadratic Taylor equation for the scalar $\hat\theta$, and there is no quadratic formula in $\mathbb{R}^{k}$. Nor does the obvious repair work. Linearizing, $\hat\theta-\theta_{0}=-\bar U^{-1}\bar\psi+\rho_{n}$ with $\bar\psi:=n^{-1}\sum_{i}\psi_{i}(\theta_{0})$ and $\bar U:=n^{-1}\sum_{i}\nabla_{\theta}\psi_{i}(\theta_{0})$, the quadratic remainder $\rho_{n}$ is itself of order $n^{-1}$, hence of order $n^{-1/2}$ on the standardized scale---exactly the accuracy at stake. This borderline term is what the univariate bracketing absorbs so carefully, and in $\mathbb{R}^{k}$ we absorb it instead into a higher-order smooth proxy, at the price of two more derivatives. Assume, in place of (M1) and (M4),
\begin{enumerate}
	\item[(M1$'$)] $m_{\theta}(x)$ is five times continuously differentiable in $\theta$ on $\bar B(\theta_{0},\delta)$;
	\item[(M4$'$)] $\E\norm{\nabla^{j}\psi(\theta_{0})}^{3}<\infty$ for $j\le3$, and $\E\sup_{\bar B(\theta_{0},\delta)}\norm{\nabla^{4}\psi}^{3}<\infty$.
\end{enumerate}
Expanding the estimating equation to third order around $\theta_{0}$,
\[
0\;=\;\bar\psi+\bar U\,d+\tfrac12\,\bar W_{2}[d,d]+\tfrac16\,\bar W_{3}[d,d,d]+r_{4},
\qquad d:=\hat\theta-\theta_{0},\quad \norm{r_{4}}\le\tfrac1{24}\,\bar R^{*}\abs{d}^{4},
\]
where $\bar W_{j}$ is the averaged $j$-th derivative tensor of $\psi$ at $\theta_{0}$ and $\bar R^{*}:=n^{-1}\sum_{i}\sup_{\bar B(\theta_{0},\delta)}\norm{\nabla^{4}\psi_{i}}$. Dropping $r_{4}$ leaves a polynomial system in $d$ whose coefficients are the sample means of the i.i.d.\ arrays $V_{i}:=\big(\psi_{i},\nabla\psi_{i},\nabla^{2}\psi_{i},\nabla^{3}\psi_{i}\big)(\theta_{0})$. Write $v^{*}:=\E V_{1}=\big(0,\,I_{2},\,\E\nabla^{2}\psi_{1},\,\E\nabla^{3}\psi_{1}\big)$, let $\lambda_{0}:=\lambda_{\min}(I_{2})>0$, and for $v=(s,U,W_{2},W_{3})$ and $t\in\mathbb{R}^{k}$ define the polynomial map
\[
F(t;v)\;:=\;s+U\,t+\tfrac12\,W_{2}[t,t]+\tfrac16\,W_{3}[t,t,t],
\]
so that the expansion above reads $F(d;\bar V)=-r_{4}$, while $F(0;\bar V)=\bar\psi$. The argument runs through four lemmas: an implicit solution map for the polynomial system, with a quantitative injectivity estimate (Lemma~\ref{lem:mv-ift}); the delta-method bound of Pinelis and Molzon applied to that map (Lemma~\ref{lem:mv-proxy}); a good event of probability $1-O(n^{-3/2})$ (Lemma~\ref{lem:mv-events}); and the conversion of the Taylor defect $r_{4}$ into a distance between $\hat\theta$ and the proxy (Lemma~\ref{lem:mv-defect}). Theorem~\ref{thm:mv-be} assembles them.

\begin{lemma}[Implicit solution map]\label{lem:mv-ift}
There exist $\eps_{0}>0$, $\rho_{0}\in(0,1]$ and a map $H:\bar B(v^{*},\eps_{0})\to\mathbb{R}^{k}$, infinitely differentiable on a neighborhood of $\bar B(v^{*},\eps_{0})$, such that $H(v^{*})=0$, $F(H(v);v)=0$ and $\abs{H(v)}\le\rho_{0}/2$ for all $v\in\bar B(v^{*},\eps_{0})$, and
\begin{equation}\label{eq:mv-jac}
\norm*{\partial_{t}F(t;v)-I_{2}}\;\le\;\tfrac{\lambda_{0}}{2}
\qquad\text{for all }\ \abs{t}\le\rho_{0},\ \norm{v-v^{*}}\le\eps_{0}.
\end{equation}
Consequently, for every such $v$ and all $t_{1},t_{2}\in\bar B(0,\rho_{0})$,
\begin{equation}\label{eq:mv-inj}
\abs*{F(t_{1};v)-F(t_{2};v)}\;\ge\;\tfrac{\lambda_{0}}{2}\,\abs{t_{1}-t_{2}} .
\end{equation}
The differential of $H$ at $v^{*}$ acts on the $\psi$-block as $-I_{2}^{-1}$ and annihilates the remaining blocks.
\end{lemma}

\begin{proof}
$F$ is polynomial in $(t,v)$, $F(0;v^{*})=\E\psi(\theta_{0})=0$, and $\partial_{t}F(0;v^{*})=I_{2}$ is nonsingular by (M2), so the implicit function theorem yields a $C^{\infty}$ solution map $H$ on a neighborhood of $v^{*}$ with $H(v^{*})=0$ and $\nabla_{v}H(v^{*})=-\big(\partial_{t}F(0;v^{*})\big)^{-1}\partial_{v}F(0;v^{*})$; since $\partial_{v}F(0;v^{*})$ maps a direction $(s,U,W_{2},W_{3})$ to $s$, this differential is $-I_{2}^{-1}$ on the $\psi$-block and zero on the others. The map $(t,v)\mapsto\partial_{t}F(t;v)=U+W_{2}[t,\cdot\,]+\tfrac12 W_{3}[t,t,\cdot\,]$ is continuous and equals $I_{2}$ at $(0,v^{*})$, so \eqref{eq:mv-jac} holds after shrinking; shrink once more, using continuity of $H$ at $v^{*}$, so that $\sup_{\bar B(v^{*},\eps_{0})}\abs{H}\le\rho_{0}/2$. For \eqref{eq:mv-inj}, apply the mean value inequality to $t\mapsto F(t;v)-I_{2}t$ on the convex set $\bar B(0,\rho_{0})$: by \eqref{eq:mv-jac} its differential has norm at most $\lambda_{0}/2$ there, whence
$\abs{F(t_{1};v)-F(t_{2};v)}\ge\abs{I_{2}(t_{1}-t_{2})}-\tfrac{\lambda_{0}}{2}\abs{t_{1}-t_{2}}\ge\tfrac{\lambda_{0}}{2}\abs{t_{1}-t_{2}}$.
\end{proof}

\begin{lemma}[Optimal-order bound for the proxy]\label{lem:mv-proxy}
Fix a unit vector $a\in\mathbb{R}^{k}$ with $a^{\top}Va>0$, and define $f(v):=a^{\top}H(v)$ for $v\in\bar B(v^{*},\eps_{0})$ and $f(v):=0$ otherwise. Then there is a $C_{a}'<\infty$, not depending on $n$, with
\[
\sup_{z\in\mathbb{R}}\ \abs*{\Prob\paren*{\sqrt{\frac{n}{a^{\top}Va}}\;f(\bar V)\le z}-\Phi(z)}\ \le\ \frac{C_{a}'}{\sqrt{n}} .
\]
\end{lemma}

\begin{proof}
We verify the hypotheses of \cite[Theorem~3.8]{Pinelis2016} with $p=3$ for the centered i.i.d.\ vectors $V_{i}-v^{*}$ and the Borel function $x\mapsto f(v^{*}+x)$. \emph{Moments:} $\E\norm{V_{1}-v^{*}}^{3}<\infty$ by (M3) and (M4$'$). \emph{Smoothness:} condition (3.6) of \cite{Pinelis2016} requires a nonzero continuous linear functional $L$ and constants $M,\eps>0$ with $\abs{f(v^{*}+x)-L(x)}\le\tfrac{M}{2}\norm{x}^{2}$ for $\norm{x}\le\eps$. Take $\eps:=\eps_{0}$ and $L(x):=-a^{\top}I_{2}^{-1}x_{\psi}$, where $x_{\psi}$ is the $\psi$-block of $x$: on $\norm{x}\le\eps_{0}$ we have $f(v^{*}+x)=a^{\top}H(v^{*}+x)$, the differential of this map at $x=0$ is $L$ by the last claim of Lemma~\ref{lem:mv-ift}, and Taylor's theorem gives the quadratic bound with $M:=\sup_{\bar B(v^{*},\eps_{0})}\norm{\nabla^{2}\,(a^{\top}H)}$, finite because $H$ is smooth on a neighborhood of the closed ball and $\abs{a}=1$. \emph{Nondegeneracy:} the per-observation variance of the linear part is
\[
\norm*{L(V_{1}-v^{*})}_{2}^{2}\;=\;a^{\top}I_{2}^{-1}\,\E\big[\psi\psi^{\top}\big]\,I_{2}^{-1}a\;=\;a^{\top}Va\;>\;0,
\]
so the standardization in the display is exactly that of the theorem. The conclusion is the uniform bound (3.23) of \cite{Pinelis2016}.
\end{proof}

\begin{lemma}[Good event]\label{lem:mv-events}
Set $K^{*}:=(1+\E R^{*}_{1})/24$, $\delta_{1}:=\min\set{\delta,\ \rho_{0},\ (\lambda_{0}/(8K^{*}))^{1/3}}$, and
\[
G_{n}\;:=\;\set*{\norm{\bar V-v^{*}}\le\eps_{0}}\ \cap\ \set*{\bar R^{*}\le 1+\E R^{*}_{1}}\ \cap\ \set*{\abs{\hat\theta-\theta_{0}}\le\delta_{1}} .
\]
Then $\Prob(G_{n}^{\mathrm{c}})\le C\,n^{-3/2}$ for a constant $C$ not depending on $n$.
\end{lemma}

\begin{proof}
The $V_{i}-v^{*}$ are i.i.d., centered, with $\E\norm{V_{1}-v^{*}}^{3}<\infty$ by (M3) and (M4$'$); the Rosenthal-type inequality for sums of independent random vectors (\cite[(3.5)]{Pinelis2016}, with $\alpha=3$) gives $\E\norm{\sum_{i\le n}(V_{i}-v^{*})}^{3}\le Cn^{3/2}$, and Markov's inequality yields $\Prob\big(\norm{\bar V-v^{*}}>\eps_{0}\big)\le Cn^{3/2}/(n\eps_{0})^{3}=C\eps_{0}^{-3}\,n^{-3/2}$. The same bound applied to the scalar sums $\sum_{i}(R^{*}_{i}-\E R^{*}_{1})$, whose summands have finite third moment by (M4$'$), controls the second event. For the third, (A5) and (A6) hold with the Euclidean norm, and the argument of Section~\ref{sec:remainder} gives $\Prob\big(\abs{\hat\theta-\theta_{0}}>\delta_{1}\big)\le C_{1}e^{-c_{2}n}$, which is $O(n^{-3/2})$ a fortiori.
\end{proof}

\begin{lemma}[From defect to distance]\label{lem:mv-defect}
On $G_{n}$, with $d:=\hat\theta-\theta_{0}$ and $\tilde\theta:=\theta_{0}+H(\bar V)$,
\[
\abs{d}\;\le\;\frac{4}{\lambda_{0}}\,\norm{\bar\psi}
\qquad\text{and}\qquad
\abs*{\hat\theta-\tilde\theta}\;\le\;\frac{2K^{*}}{\lambda_{0}}\paren*{\frac{4}{\lambda_{0}}}^{4}\norm{\bar\psi}^{4} .
\]
\end{lemma}

\begin{proof}
On $G_{n}$ the expansion $F(d;\bar V)=-r_{4}$ holds with $\norm{r_{4}}\le\tfrac1{24}\bar R^{*}\abs{d}^{4}\le K^{*}\abs{d}^{4}$, and $\abs{d}\le\delta_{1}\le\rho_{0}$, $\norm{\bar V-v^{*}}\le\eps_{0}$, so the injectivity estimate \eqref{eq:mv-inj} is available on $\bar B(0,\rho_{0})$. Comparing $d$ with $0$, and using $F(0;\bar V)=\bar\psi$ together with $\abs{d}^{4}\le\delta_{1}^{3}\abs{d}$ and $K^{*}\delta_{1}^{3}\le\lambda_{0}/8$,
\[
\tfrac{\lambda_{0}}{2}\,\abs{d}\;\le\;\abs*{F(d;\bar V)-F(0;\bar V)}\;=\;\norm*{r_{4}+\bar\psi}\;\le\;\norm{\bar\psi}+K^{*}\delta_{1}^{3}\abs{d}\;\le\;\norm{\bar\psi}+\tfrac{\lambda_{0}}{4}\abs{d},
\]
which rearranges to the first claim. Comparing $d$ with $H(\bar V)$---both lie in $\bar B(0,\rho_{0})$, the former because $\delta_{1}\le\rho_{0}$ and the latter by Lemma~\ref{lem:mv-ift}---and using $F(H(\bar V);\bar V)=0$,
\[
\tfrac{\lambda_{0}}{2}\,\abs*{d-H(\bar V)}\;\le\;\abs*{F(d;\bar V)-F(H(\bar V);\bar V)}\;=\;\norm{r_{4}}\;\le\;K^{*}\abs{d}^{4}\;\le\;K^{*}\paren*{\frac{4}{\lambda_{0}}}^{4}\norm{\bar\psi}^{4},
\]
by the first claim, which rearranges to the second.
\end{proof}

\begin{theorem}[Directional Berry--Esseen bound]\label{thm:mv-be}
Assume (M1$'$), (M2), (M3), (M4$'$) and (A5)--(A6). For every unit vector $a\in\mathbb{R}^{k}$ with $a^{\top}Va>0$ there is a constant $C_{a}<\infty$, depending on the direction only through $a^{\top}Va$ and otherwise on the moments and constants in the assumptions but not on $n$, such that
\[
\sup_{z\in\mathbb{R}}\ \abs*{\Prob\paren*{\sqrt{\frac{n}{a^{\top}Va}}\;a^{\top}\big(\hat\theta-\theta_{0}\big)\le z}-\Phi(z)}\ \le\ \frac{C_{a}}{\sqrt{n}} .
\]
\end{theorem}

\begin{proof}
Write $T:=\sqrt{n/(a^{\top}Va)}\;a^{\top}d$ and $T_{f}:=\sqrt{n/(a^{\top}Va)}\;f(\bar V)$ with $f$ as in Lemma~\ref{lem:mv-proxy}, and set $\Delta:=T-T_{f}$. For any $\eps>0$,
\[
\sup_{z}\abs*{\Prob(T\le z)-\Phi(z)}\ \le\ \sup_{z}\abs*{\Prob(T_{f}\le z)-\Phi(z)}\;+\;\Prob\big(\abs{\Delta}>\eps\big)\;+\;\frac{\eps}{\sqrt{2\pi}},
\]
the last term because $\Phi$ is Lipschitz with constant $1/\sqrt{2\pi}$. Take $\eps=n^{-1/2}$. The first term is at most $C_{a}'/\sqrt{n}$ by Lemma~\ref{lem:mv-proxy}. For the middle term: on $G_{n}$ we have $\norm{\bar V-v^{*}}\le\eps_{0}$, hence $f(\bar V)=a^{\top}H(\bar V)=a^{\top}(\tilde\theta-\theta_{0})$, and Lemma~\ref{lem:mv-defect} with $\abs{a}=1$ gives
\[
\abs{\Delta}\;\le\;\sqrt{\frac{n}{a^{\top}Va}}\;\abs*{\hat\theta-\tilde\theta}\;\le\;c_{*}\,\big(\sqrt{n}\,\norm{\bar\psi}\big)^{4}\,n^{-3/2},
\qquad
c_{*}:=\frac{2K^{*}}{\lambda_{0}}\paren*{\frac{4}{\lambda_{0}}}^{4}\frac{1}{\sqrt{a^{\top}Va}} .
\]
Therefore, by Lemma~\ref{lem:mv-events} and Markov's inequality applied at the third power,
\[
\Prob\big(\abs{\Delta}>n^{-1/2}\big)\ \le\ \Prob\big(G_{n}^{\mathrm{c}}\big)+\Prob\Big(\big(\sqrt{n}\norm{\bar\psi}\big)^{4}>n/c_{*}\Big)
\ \le\ C\,n^{-3/2}+\E\big(\sqrt{n}\norm{\bar\psi}\big)^{3}\,c_{*}^{3/4}\,n^{-3/4},
\]
and $\E\big(\sqrt{n}\norm{\bar\psi}\big)^{3}$ is bounded uniformly in $n$ by the Rosenthal-type inequality of Lemma~\ref{lem:mv-events}, since $\E\norm{\psi(\theta_{0})}^{3}<\infty$ by (M3). Every contribution is $O(n^{-1/2})$; collecting constants gives $C_{a}$.
\end{proof}

\begin{remark}
The smoothness requirements of the two arguments differ: the univariate bracketing of Section~\ref{sec:remainder} needs three derivatives of the criterion, the proxy route five. The cubic proxy yields the quartic defect required for our proof. We do not claim that this smoothness is necessary; it is sufficient for the argument presented here. There is no quadratic formula in $\mathbb{R}^{k}$ to bracket with, and the first-order (Newton) proxy leaves a quadratic defect of exactly the order $n^{-1/2}$ under scrutiny; the third-order proxy leaves a quartic defect, which Lemma~\ref{lem:mv-defect} and a third-moment bound push strictly below it.
\end{remark}

For the reporting-scale estimator of Section~\ref{sec:scale}, conditions (M1)--(M4)---and equally (M1$'$) and (M4$'$)---with $m_{(\theta,c)}=c\,p_{\theta}$ reduce to their univariate counterparts for the underlying family: all derivatives in $c$ beyond the first vanish ($\partial_{c}m=p_{\theta}$, $\partial_{c}^{2}m=0$), so each condition is implied by the corresponding moment condition on $p_{\theta}$ and its $\theta$-derivatives, and Theorem~\ref{thm:mv-be} applies to $(\hat\theta,\hat c)$. The finite-sample guarantee anticipated there is therefore not an extra assumption but a corollary of this section.

We validate the multivariate theory on the elicitation problem it is actually for: recovering a location \emph{and} a scale at once. Take $p_{(\mu,s)}(x)=\varphi\big((x-\mu)/s\big)/s$ with $\theta_{0}=(\mu_{0},s_{0})=(2,1.5)$, $\sigma=0.1$, and the least-squares fit over the box $[-3,7]\times[0.3,5]$; the family is infinitely differentiable in $(\mu,s)$ with Gaussian envelopes, so (M1$'$)--(M4$'$) hold with room to spare. At $\theta_{0}$ the gradient components $\partial_{\mu}p=p\,u/s$ and $\partial_{s}p=p\,(u^{2}-1)/s$ (with $u=(x-\mu_{0})/s_{0}$) are odd and even in $u$ respectively, so both $I_{1}$ and $I_{2}$ are diagonal and the sandwich $V$ is diagonal as well: location and scale are estimated asymptotically independently, with $V=\operatorname{diag}(0.0543,\ 0.0489)$ at $\sigma=0.1$. Simulation over $4{,}000$ replications per sample size confirms every layer (Figure~\ref{fig:ch5-multivariate}): the per-coordinate mean squared errors decay with fitted log--log slopes $-1.03$ ($\mu$) and $-1.02$ ($s$) over $n=25,\ldots,400$; at $n=200$ the empirical $n\,\mathrm{Cov}$ matrix is $\big(\begin{smallmatrix}0.0546&-0.0006\\-0.0006&0.0495\end{smallmatrix}\big)$ against the predicted diagonal above; and the errors standardized by $\sqrt{a^{\top}Va/n}$ along the directions $e_{1}$, $e_{2}$ and $(e_{1}+e_{2})/\sqrt2$ pass Kolmogorov--Smirnov tests against $\mathcal{N}(0,1)$ with statistics $0.010$--$0.012$ ($p$-values $0.58$--$0.84$). The finite-sample bound itself is confirmed in the only sense a simulation of this size can resolve: from $n\approx50$ onward the measured Kolmogorov distances sit at the resolution floor of $4{,}000$ replications ($\approx0.014$), so the distance to normality is already below measurement precision at sample sizes an elicitation session would actually use---the $n^{-1/2}$ decay predicted by the bound cannot be distinguished because there is nothing left to decay.

\begin{figure}
	\centering
	\includegraphics[width=\linewidth]{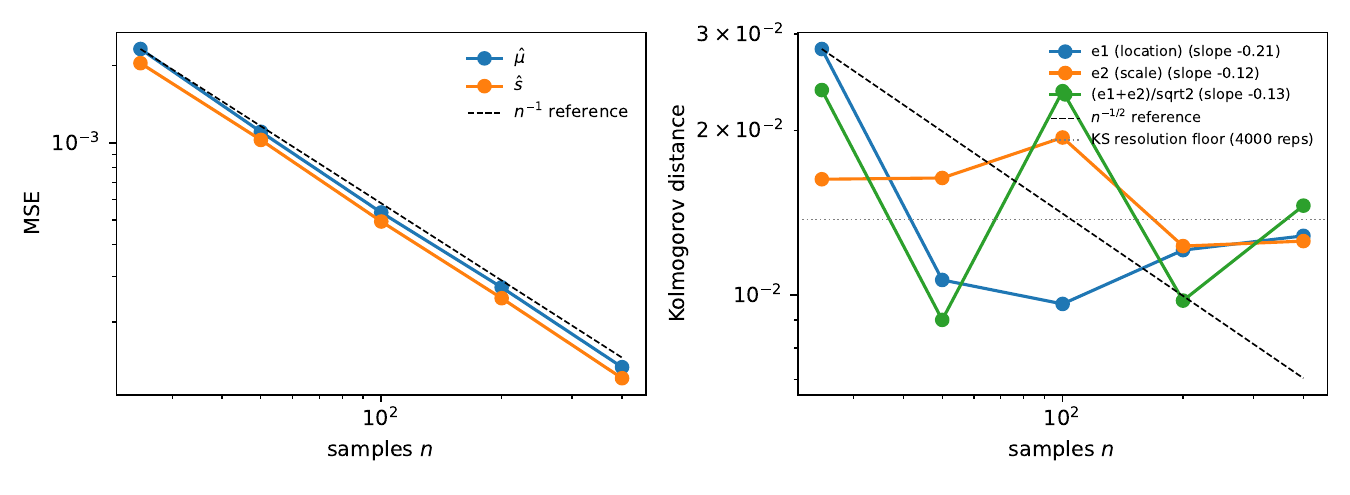}
	\caption[Multivariate elicitation: normal location--scale family.]{Joint elicitation of $(\mu,s)$ for $\mathcal{N}(\mu,s^{2})$ at $\theta_{0}=(2,1.5)$, $\sigma=0.1$, $4{,}000$ replications per point. Left: per-coordinate mean squared error against $n$ (log--log), with fitted slopes $-1.03$ and $-1.02$ against the $n^{-1}$ reference. Right: Kolmogorov distance of the standardized directional errors to $\mathcal{N}(0,1)$; from $n\approx50$ the distances sit at the resolution floor of the replication budget (dotted), so convergence to normality is complete within measurement precision.}
	\label{fig:ch5-multivariate}
\end{figure}

\section{Experiments}
\label{sec:elicitation-experiments}

To validate our learnability results, we performed experiments using synthetic data. Throughout, we simulate the batch data that would be obtained from an elicitation by sampling a ``target'' distribution $p_{\theta_0}$ and reporting sample--likelihood pairs $(x_i, z_i)$ with $z_i = p_{\theta_0}(x_i)$. To model an expert who can give realistic samples but may be variably good at assessing their plausibility, we corrupt the reported likelihoods multiplicatively, $z_i = p_{\theta_0}(x_i)(1+\sigma\xi_i)$ with $\xi_i\sim\mathcal{N}(0,1)$, where $\sigma$ controls the expert's reliability. We then learn $\theta$ by minimizing the least-squares objective of Section~\ref{sec:objective} (a global grid search over a bounded parameter interval followed by local refinement, so that reported failures reflect the information in the data rather than optimization artifacts), and assess the quality of the elicitation by the squared error of the learned parameter and the KL divergence of the learned distribution from the target. We use two univariate families: the location family $\mathcal{N}(\theta,1)$ with $\theta_0=2$, searched over $\Theta=[-3,7]$, and the shape family $\mathrm{Beta}(\theta,2)$ with $\theta_0=3$, searched over $\Theta=[1.1,10]$---in each case the parameter set on which the uniform-concentration condition (A6) is verified in Section~\ref{sec:appendix-a6}, so the guarantees apply to the estimator exactly as implemented.

In the absence of reporting noise, the population criterion is minimized uniquely at $\theta_0$ under identifiability on the sampling support. The empirical objective also vanishes at $\theta_0$, but exact finite-sample recovery additionally requires that the realized density evaluations identify the parameter: $p_\theta(x_i)=p_{\theta_0}(x_i)$ for every $i$ must imply $\theta=\theta_0$. The statistically interesting regime---and the one covered by the asymptotic theory, whose limiting variance is $I_1(\theta_0)/(nI_2(\theta_0)^2)$ with $I_1$ driven by the noise---is $\sigma>0$, which is also the realistic description of a human expert.

\subsection{Validating the estimator and the learning rate}

We vary the number of samples per batch from 2 to 40 and run 200 batches per configuration, declaring a batch a success when $|\hat\theta-\theta_0|<0.25$. Figure \ref{fig:elicitation-success} shows that the success rate is directly related to the sample size at every noise level, and degrades gracefully with $\sigma$: for the normal family, even the noisiest expert ($\sigma=0.2$) is elicited successfully $96\%$ of the time by $n=8$ samples, while the harder beta shape family requires roughly $40$ samples at that noise level ($34\%$ success at $n=2$, rising to $96\%$ at $n=40$). Figure \ref{fig:elicitation-rate} shows the mean squared parameter error against the sample size on log--log axes. For $n\gtrsim 8$ the fitted log--log slopes lie in $[-1.19,-0.95]$ across both families and all three noise levels, matching the $n^{-1}$ rate predicted by the asymptotic normality result. Over the full plotted range $n=2,\ldots,40$ the fitted slopes are steeper, between $-1.89$ and $-1.25$: at the smallest batch sizes the estimator is still far from its asymptotic regime, and the error falls faster than $n^{-1}$ before settling onto the predicted rate. The KL divergence behaves identically---for the normal location family it equals $(\hat\theta-\theta_0)^2/2$ exactly, and for the beta family it is locally quadratic in the parameter error---so we do not plot it separately.

\begin{figure}
	\centering
	\includegraphics[width=0.48\linewidth]{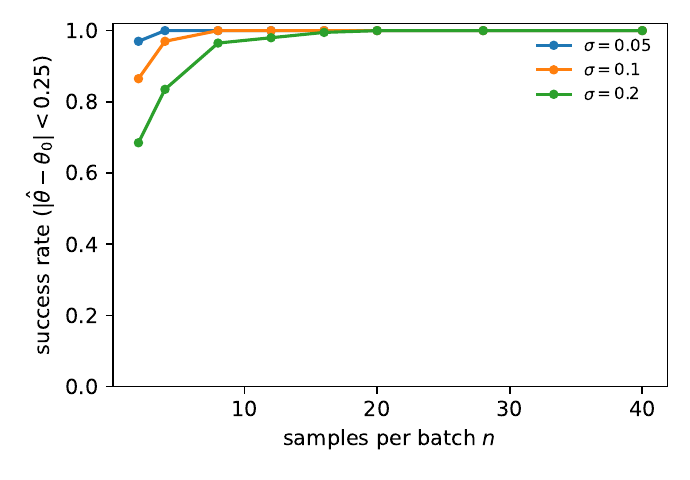}
	\hfill
	\includegraphics[width=0.48\linewidth]{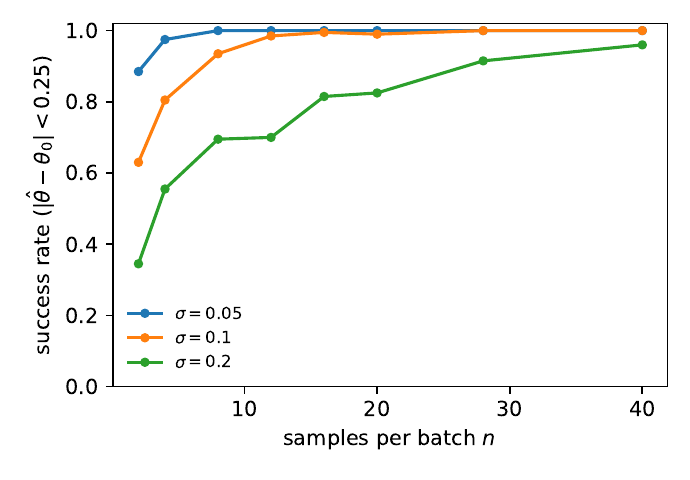}
	\caption[Elicitation success rate versus samples per batch.]{Success rate of elicitation ($|\hat\theta-\theta_0|<0.25$) versus samples per batch over 200 batches, at three levels of expert noise. Left: $\mathcal{N}(\theta,1)$, $\theta_0=2$. Right: $\mathrm{Beta}(\theta,2)$, $\theta_0=3$.}
	\label{fig:elicitation-success}
\end{figure}

\begin{figure}
	\centering
	\includegraphics[width=0.48\linewidth]{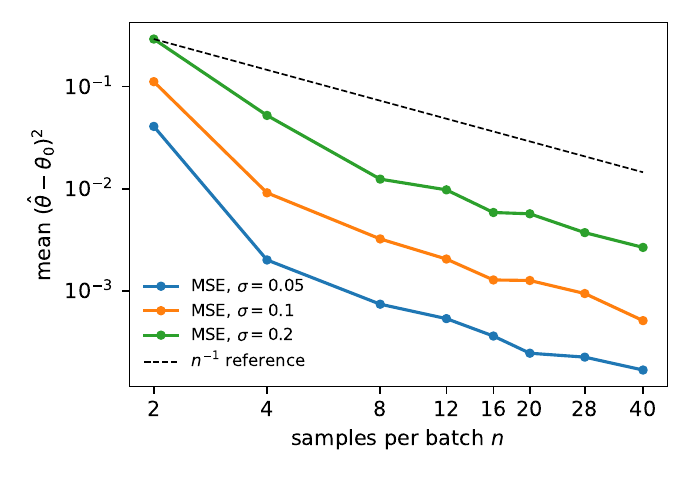}
	\hfill
	\includegraphics[width=0.48\linewidth]{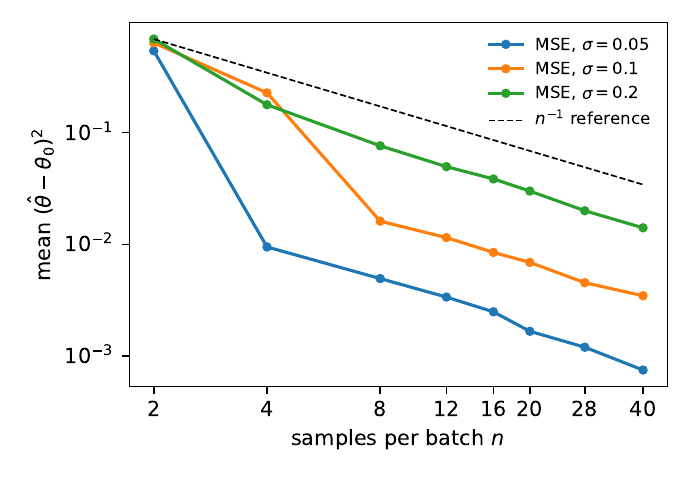}
	\caption[Mean squared error of the elicited parameter.]{Mean squared error of the elicited parameter versus samples per batch (log--log), with an $n^{-1}$ reference line. Left: normal location family. Right: beta shape family.}
	\label{fig:elicitation-rate}
\end{figure}

\subsection{Validating asymptotic normality}

The finite-sample theory developed below asserts that $\sqrt{n I_2(\theta_0)^2/I_1(\theta_0)}\,(\hat\theta-\theta_0)$ is close to standard normal, with Kolmogorov distance decaying as $C/\sqrt{n}$. We check this directly: for the normal family with $\sigma=0.1$ and $n=200$, we compute the standardized error over $2{,}000$ replications, estimating $I_1(\theta_0)=\E\,\psi^2((X,Z),\theta_0)$ and $I_2(\theta_0)$ by Monte Carlo. The resulting sample has mean $-0.03$ and standard deviation $1.00$, and a Kolmogorov--Smirnov test against $\mathcal{N}(0,1)$ gives statistic $0.023$ ($p=0.25$): the sampling distribution is statistically indistinguishable from the theoretical limit at this sample size. Figure \ref{fig:elicitation-normality} overlays the histogram on the standard normal density.

\begin{figure}
	\centering
	\includegraphics[width=0.6\linewidth]{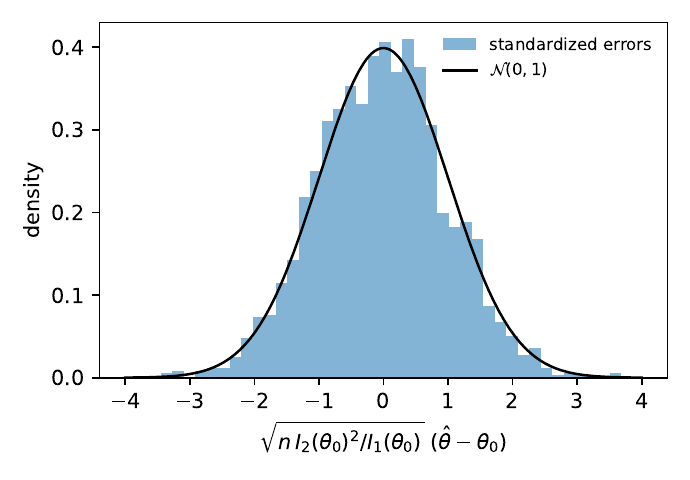}
	\caption[Sampling distribution of the standardized elicitation error.]{Histogram of the standardized error $\sqrt{n I_2(\theta_0)^2/I_1(\theta_0)}\,(\hat\theta-\theta_0)$ over $2{,}000$ elicitation replications (normal family, $n=200$, $\sigma=0.1$), with the $\mathcal{N}(0,1)$ density overlaid.}
	\label{fig:elicitation-normality}
\end{figure}

\subsection{A family with no moments}
\label{sec:cauchy}

The objective is defined by the density alone, so it remains well posed for families for which moment matching is undefined. We take the Cauchy location family $p_{\theta}(x)=[\pi(1+(x-\theta)^{2})]^{-1}$ with $\theta_0=2$, which has no mean and no finite absolute moment of order $\ge1$ (its absolute moments of order $0\le p<1$ are finite, equal to $\sec(\pi p/2)$ in the standard case). The sample mean therefore does not converge: over five independent batches of $n=40$ it took the values $-0.11$, $4.14$, $2.19$, $2.49$ and $2.10$. A method-of-moments fit has nothing to match. The least-squares elicitation estimator is unaffected. Over $300$ replications it recovers $\theta_0$ to within $0.25$ in every replication at $\sigma\in\set{0.05,0.1}$ for $n=10,20,40$, and at $\sigma=0.2$ its success rate rises from $93\%$ at $n=10$ to $99\%$ at $n=20$ and $100\%$ at $n=40$. Regressing $\log$ mean squared error on $\log n$ over $n=10,\ldots,80$ at $\sigma=0.1$ gives a slope of $-1.16$, consistent with the $n^{-1}$ rate seen for the normal and beta families.

\subsection{How noisy are real reported magnitudes? A semi-synthetic check}
\label{sec:semisynthetic}

The crossover $\sigma^{\ast}$ of Section~\ref{sec:ls-vs-mle} is only useful if real reporting noise can fall below it, and no public dataset of density elicitation with known ground truth exists against which to check this. The closest well-replicated task is the judgment of annual death frequencies: participants state the number of deaths per year for each of up to $41$ causes, and the true frequencies are known. Pachur \cite{Pachur2024} collated the original study of Lichtenstein et al.\ \cite{Lichtenstein1978} with its replications---eleven datasets from eight studies spanning $1978$--$2020$ and three countries, each reporting the aggregate (geometric-mean or median) judged frequency per cause. We fit the chapter's reporting models to each dataset. Under the scale-only model $z=c\,p\,(1+\sigma\xi)$ of Section~\ref{sec:scale} the implied relative error is $2.9$--$13.3$: the dominant deviation is not noise at all but \emph{compression}, the classic primary bias---regressing $\log z$ on $\log p$ gives slopes $b$ between $0.42$ and $0.72$ (median $0.48$) rather than $1$. Allowing the compression, $z=c\,p^{b}(1+\sigma\xi)$, the residual relative error is $\hat\sigma=0.79$--$1.97$: \emph{above the crossover in every dataset} (Figure~\ref{fig:ch5-semisynthetic}, left). For reports of absolute magnitudes, then, the verdict is negative---a sample-only maximum-likelihood fit would beat plausibility-weighted least squares on this task, and since these are aggregates over $39$--$85$ participants, individual reporting noise is higher still. The mitigating consideration cuts the other way: judging absolute frequencies of $41$ disparate causes across five orders of magnitude is a recall task about the world, not the local, relative judgment about one's own belief that the elicitation instrument requests, so these values are better read as an upper bound on elicitation-type reporting noise. Direct measurement of $\sigma$ for relative-plausibility reports is the human-subject study that remains open.

What can be elicited \emph{at} these noise levels? We simulate a calibrated expert, $z=c\,p_{\theta_0}(x)^{b}(1+\sigma\xi)$ with the empirically fitted $(b,\sigma)$---a reporting process that includes the compression our estimator does not model---and run the joint $(\theta,c)$ fit of Section~\ref{sec:scale}. For the symmetric location family, compression is a symmetric widening ($p^{b}$ is proportional to a density with the same center), so the location estimate remains \emph{unbiased} even at the worst-case calibration $(b,\sigma)=(0.5,1.0)$: across $n=8$ to $80$ the absolute bias never exceeds $0.04$, and the success rate climbs from $0.25$ to $0.72$ ($0.81$ under the best-dataset calibration $(0.45,0.79)$; Figure~\ref{fig:ch5-semisynthetic}, right). The loss relative to the uncompressed model at $\sigma^{\ast}$ (success $0.97$ at $n=80$) is pure variance, curable by asking for more examples. Shape parameters are not protected: the beta fit under the same calibrations acquires an asymptotic bias of $-0.19$ to $-0.26$ and its success rate plateaus near one third. The practical reading is that elicitation degrades gracefully in exactly one direction---the \emph{location} of a belief survives even the harshest documented reporting behavior, while precision beyond location is what expert noise destroys first---and the open empirical question is sharpened accordingly: what matters is $\sigma$ for local relative judgments, not for absolute magnitudes.

\begin{figure}
	\centering
	\includegraphics[width=\linewidth]{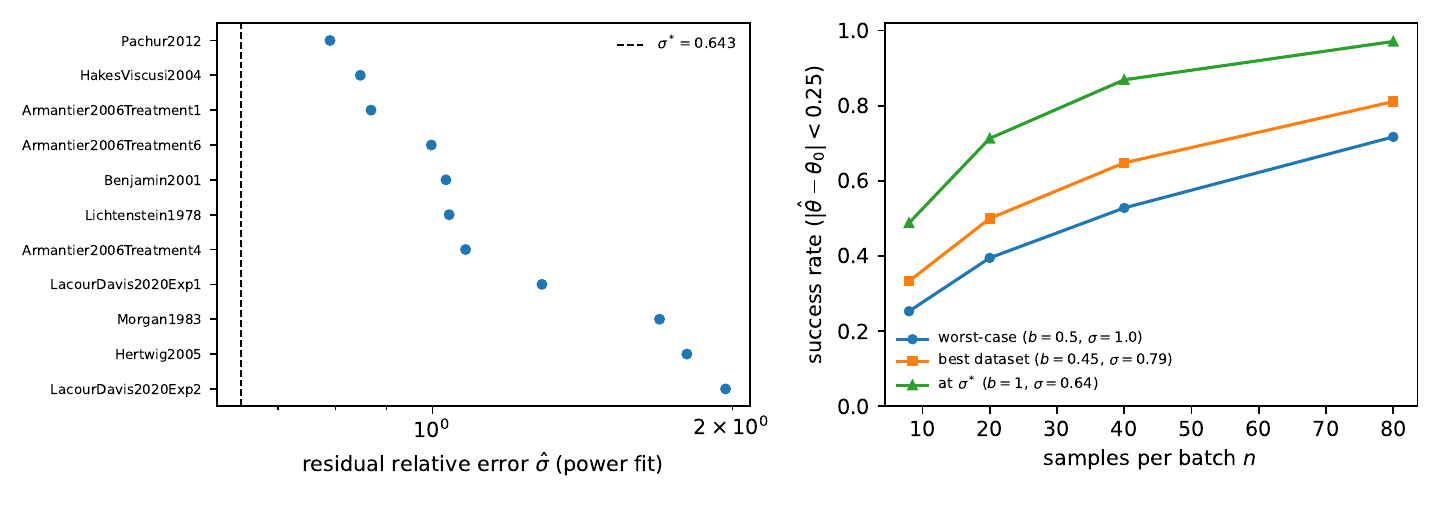}
	\caption[Reporting noise in real frequency judgments and calibrated elicitation.]{Left: residual relative error $\hat\sigma$ of aggregate human frequency judgments in the eleven datasets collated by \cite{Pachur2024}, after fitting the power reporting model $z=c\,p^{b}(1+\sigma\xi)$ per dataset (log scale); every dataset lies above the crossover $\sigma^{\ast}$ (dashed). Right: success rate of the joint $(\theta,c)$ estimator for the normal location family when the expert is simulated with empirically calibrated compression and noise; the location estimate is unbiased in all three conditions, and the gap between the curves is variance only.}
	\label{fig:ch5-semisynthetic}
\end{figure}

\section{Proof of Theoretical Bound}
Here, we provide a theoretical proof of the main asymptotic result. This follows the technique introduced in \cite{Pinelis2017} for maximum likelihood estimators. Pinelis briefly states that his result could be extended to M-estimators, and in the following, we fully exposit the proof for the general class of M-estimators, which requires adjustments to the assumptions in \cite{Pinelis2017}.

Finite-sample normal approximations for M-estimators have a long history, and we situate the result below within it. Pfanzagl \cite{Pfanzagl1971,Pfanzagl1972} and Michel and Pfanzagl \cite{michel1971accuracy} obtained $O(n^{-1/2})$ bounds for minimum contrast estimates under a twice-differentiable score with Lipschitz second derivative, extended to vector parameters in \cite{Pfanzagl1973}; Paulauskas \cite{Paulauskas1996} treats the multivariate case under smooth stochastic differentiability; and Bentkus, Bloznelis and G\"otze \cite{BentkusBloznelisGotze1997} give an explicit $O(n^{-1/2})$ bound for general one-dimensional M-estimators, assuming that the sample criterion is convex on a neighborhood of $\theta_{0}$ with probability $1-O(n^{-1/2})$, together with a consistency hypothesis of the same order; their Corollary~1.3 dispenses with the convexity requirement for twice-differentiable criteria. Our criterion falls within the scope of the last of these. Although $\ell_{x,z}(\theta)=-(p_{\theta}(x)-z)^{2}$ is concave for no family with $\sup_{\theta}p_{\theta}(x)<\infty$ (Remark~\ref{rem:concavity}), it is smooth, and the separation $D$ of Section~\ref{sec:consistency} has $D''(\theta_{0})>0$ for both families used here, so the sample criterion is locally convex near $\theta_{0}$ with high probability. What the present analysis contributes is therefore not the bound itself but the hypotheses it runs on: conditions (A5) and (A6) are stated in closed form, are verified in Section~\ref{sec:appendix-a6} for the families used here, and \emph{yield} the consistency that \cite{BentkusBloznelisGotze1997} assumes---with an exponential rate in place of $O(n^{-1/2})$, and, for (A5) in general and (A6) for the normal location family, with no compactness requirement on $\Theta$ (Proposition~\ref{prop:a6-unbounded}). The criterion treated here is also outside the worked examples of \cite{BentkusBloznelisGotze1997}, which are the sample median, the $L_{p}$-median and Huber's estimator of location.

For the multivariate bound of Section~\ref{sec:multivariate}, the closest work is Shao and Zhang \cite{ShaoZhang2021}, whose Theorem~3.1 is uniform over convex sets and therefore stronger than the directional statement given here. Their condition~(M1) is not a convexity hypothesis but a global quadratic growth condition, $M(\theta)-M(\theta_{0})\ge\mu\norm{\theta-\theta_{0}}^{2}$ on all of $\Theta$, and on a bounded parameter set it holds in our setting: the separation $D(\theta)$ of Section~\ref{sec:consistency} gives $\mu\approx3.7\times10^{-3}$ on $\Theta=[-3,7]$ for the normal location family and $\mu\approx4.4\times10^{-2}$ on $\Theta=[1.1,10]$ for the beta shape family. Their Theorem~3.2 for $Z$-estimators does not apply, since it requires the population estimating function to be globally strongly monotone, and that function is $D'/2$, which is not monotone on either interval---a consequence of the same saturation of $D$ recorded in Remark~\ref{rem:concavity}. On an unbounded parameter set that saturation defeats (M1) as well, for every $\mu>0$, while (A5) continues to hold, and so does (A6) for the normal location family (Proposition~\ref{prop:a6-unbounded}); for the beta shape family (A6) fails on $[1.1,\infty)$, yet exponential consistency survives there by the one-sided argument of Proposition~\ref{prop:consistency-unbounded}.

This proof is organized as follows.

\begin{enumerate}
	\item We describe the general problem setting and assumptions required.
	\item We demonstrate tight bracketing of our M-estimator between two functions of the sum of independent random vectors.
	\item We present uniform and nonuniform optimal-order bounds on the convergence rate in the multivariate delta method \cite{Pinelis2016}.
	\item We apply the general bounds in the multivariate delta method such that we can make bracketing work.
	\item We bound the remainder and show this is asymptotically negligible under certain conditions.
\end{enumerate}

\subsection{Setting and Assumptions}\label{sec:setting}

Let $X, X_1, X_2, \ldots$ be random variables mapping from $(\Omega, \mathcal{A})$ to $(\mathcal{X},\mathcal{B})$ and let $(P_{\theta})_{\theta\in \Theta}$ be a parametric family of probability measures such that $X, X_1, X_2, \ldots$ are i.i.d. with respect to each of the measures $P_{\theta}$ with $\theta \in \Theta$. In this section $\Theta\subseteq \mathbb{R}$, i.e.\ the parameter space is a subset of the real line: the bracketing device below solves a scalar quadratic equation and is genuinely one-dimensional. Section~\ref{sec:multivariate} lifts the restriction by a different route.

Let $E_{\theta}$ be the expectation with respect to $P_{\theta}$. For each $\theta\in \Theta$, $P_{\theta} X^{-1}$ of $X$ has a density $p_{\theta}$ with respect to a measure $\mu$ on $\mathcal{B}$.

Because the extended real line $[-\infty,\infty]$ is compact, for each $n\in \mathbb{N}$ and point $x=x_n=(x_1,\ldots,x_n)\in \mathcal{X}^n$, the sample criterion $\Theta\ni \theta\mapsto L_n(\theta) = \sum_{i=1}^n -(l_{x_i}(\theta)-z_i)^2 $ has at least one generalized maximizer $\hat{\theta}_n(x)$ in the closure of $\Theta$. Throughout, $\ell_{x,z}$ denotes the per-observation criterion and $L_n$ the sum; the assumptions below are stated for $\ell_{X,Z}$.

Let $\theta_{0}\in \Theta$ be the expert's target value of the parameter $\theta$, such that
\begin{center}
	$[\theta_{0}-\delta,\ \theta_{0}+\delta]\subseteq \theta^{\circ}$ 
\end{center}
for some real $\delta>0$, where $\theta^{\circ}$ denotes the interior of the subset $\Theta$ of $\mathbb{R}$.

For convenience, we provide the assumptions for $\ell$ again below.

\begin{enumerate}
	\item The set $\mathcal{X}_{>0} :=\{x\in \mathcal{X}:p_{\theta}(x)>0\}$ is the same for all $\theta\in [\theta_{0}-\delta,\ \theta_{0}+\delta],$ and for each $x\in \mathcal{X}_{>0}$ $\ell_{x}(\theta)$ is thrice differentiable in $\theta$ at each point $\theta\in [\theta_{0}-\delta,\ \theta_{0}+\delta].$
	\item $\mathrm{E}\ell_{X,Z}'(\theta_0)=0$, $0<\mathrm{E}[\ell_{X,Z}'(\theta_0)^2]=I_1(\theta_0)<\infty$, and $0<-\mathrm{E}\ell_{X,Z}''(\theta_0)=I_2(\theta_0)<\infty$.
	\item $\mathrm{E}|\ell_{X,Z}'(\theta_{0})|^{3}+\mathrm{E}|\ell_{X,Z}^{''}(\theta_{0})|^{3}<\infty.$
	\item $\mathrm{E} \sup |\ell_{X,Z}^{'''}(\theta)|^{3}<\infty.$
\end{enumerate}

\medskip

\subsection{Tight Bracketing}

Without loss of generality (w.l.o.g.), $\mathcal{X}_{>0}=\mathcal{X}$. Then on the event
\begin{equation}\label{eqn:goodevent}
	G:=\{\hat{\theta}\in[\theta_{0}-\delta,\ \theta_{0}+\delta]\}
\end{equation}
($G$ for ``good event," one must have

\begin{equation}\label{eqn:taylor}
	0 =\ell_{\mathrm{x}}'(\displaystyle \hat{\theta})=\ell_{\mathrm{x}}'(\theta_{0})+(\hat{\theta}-\theta_{0})\ell_{\mathrm{x}}''(\theta_{0})+\frac{(\hat{\theta}-\theta_{0})^{2}}{2}\ell_{\mathrm{x}}'''(\theta_{0}+\xi(\hat{\theta}-\theta_{0}))
\end{equation}

\begin{equation}\label{eqn:taylor2}
	=n(\overline{Z}-(\hat{\theta}-\theta_{0})\overline{U}+\frac{(\hat{\theta}-\theta_{0})^{2}}{2}\overline{R})
\end{equation}

for some $\xi\in (0,1)$ as a function of the $X_i$'s, where $\overline{Z}=\frac{1}{n}\sum_{i=1}^n Z_i,\, \overline{U}=\frac{1}{n}\sum_{i=1}^{n}U_{i},\, \overline{R}:=\frac{1}{n}\sum_{i=1}^{n}R_{i},\,
\overline{R^{*}}:=\frac{1}{n}\sum_{i=1}^{n}R_{i}^{*},$

\begin{equation}\label{eqn:Z}
	Z_i = \ell'_{X_i}(\theta_0),\quad U_i=-\ell''_{X_i}(\theta_0)
\end{equation}

\begin{equation}\label{eqn:R}
	R_i = \ell'''_{X_i}(\theta_0 + \xi(\hat{\theta}-\theta_0))\in [-R_i^{*},R_i^{*}], \quad R_i^{*}=\sup_{\theta\in [\theta_0-\delta,\theta_0+\delta]} \abs{\ell'''_{X_i}(\theta)}.
\end{equation}

Looking at \eqref{eqn:taylor} and \eqref{eqn:taylor2}, one has a quadratic equation for $\hat{\theta}$.

On the event $G$ one has
\begin{align*}
	\hat{\theta}-\theta_{0}=\frac{\overline{Z}}{\overline{U}} &\textrm{ if } \overline{R}=0\, \&\, \overline{U}\neq 0, \\
	\hat{\theta}-\theta_{0}\in\{d_{+},\ d_{-}\} &\textrm{ if } \overline{R}\neq 0,
\end{align*}

where
$$
d_{\pm}:=\frac{\overline{U}\pm\sqrt{\overline{U}^{2}-2\overline{Z}\overline{R}}}{\overline{R}}.
$$

One defines a ``bad event" by letting

$B :=B_{1}\cup B_{2}$, where

$B_{1} :=\{\overline{R}\neq 0,\hat{\theta}-\theta_{0}=d_{+}\}\cup\{\overline{U}\leq 0\}$ and $B_{2} :=\{\overline{U}^{2}\leq 2|\overline{Z}|\overline{R^{*}}\}.$

On the event $B_1 \cap \{\overline{U}>0\}$, one sees $|\hat{\theta}-\theta_{0}|=|d_{+}|\geq \overline{U}/|\overline{R}|\geq\overline{U}/\overline{R^{*}}$

By \eqref{eqn:goodevent},
\begin{equation}\label{eqn:bound_intersection}
\Prob(G\cap B_1)\leq\Prob(\overline{U}\leq\delta\overline{R^*})
=\Prob\!\left(\sum_{i=1}^n(U_i-\delta R_i^*)\leq0\right).
\end{equation}

And by the assumptions for $\ell$ and the definitions for $Z_i$, $U_i$, $R_i$, and $R^{*}_i$,

\[
\mathrm{E}U_{1}>0,\, \mathrm{E}|Z_{1}|^{3}<\infty,\, \mathrm{E}|U_{1}|^{3}<\infty,\, \mathrm{E}(R_{1}^{*})^{3}<\infty.
\]

Therefore, $\mathrm{E}R_{1}^{*}<\infty$. Choose $\delta>0$ to be small enough such that
$$
\delta_{1}:=\mathrm{E}(U_{i}-\delta R_{i}^{*})>0.
$$
Then, letting $Y_{i} :=(U_{i}-\delta R_{i}^{*})-\mathrm{E}(U_{i}-\delta R_{i}^{*})$, we use \eqref{eqn:bound_intersection} with Markov's inequality to have

\begin{align*}
	\mathrm{P}(G\cap B_{1})\leq \mathrm{P}(\sum_{i=1}^{n}Y_{i}\leq-n\delta_{1})&\leq\frac{1}{(n\delta_{1})^{3}}\mathrm{E}|\sum_{i=1}^{n}Y_{i}|^{3}
	\\
	&\leq\frac{n\mathrm{E}|Y_{1}|^{3}+\sqrt{8/\pi}(n\mathrm{E}Y_{1}^{2})^{3/2}}{(n\delta_{1})^{3}}\leq\frac{\mathrm{C}}{n^{3/2}}
\end{align*}

where $\mathrm{C}:=(\mathrm{E}|Y_{1}|^{3}+\sqrt{8/\pi}\,(\mathrm{E}Y_{1}^{2})^{3/2})/\delta_{1}^{3}$, which depends on $\delta_{1}>0, \mathrm{E}Y_{1}^{2}<\infty,$ and $\mathrm{E}|Y_{1}|^{3}<\infty$. However, this does not depend on $n$.

Now, one notices $B_{2}$ implies at least one of the following events:
\begin{align*}
	B_{21} &=\displaystyle\{\overline{U}\leq\frac{1}{2}\mathrm{E}U_{1}\} \\
	B_{22} &=\{\overline{R^{*}}\geq 1+\mathrm{E}R_{1}^{*}\}, \textrm{ or}\\ 
	B_{23} &= \displaystyle \{|\overline{Z}|\geq\frac{1}{8}(\mathrm{E}U_{1})^{2}/(1+\mathrm{E}R_{1}^{*})\}. 
\end{align*}

So,
\begin{equation}
	\Prob(B_{2})\leq \Prob(B_{21})+\Prob(B_{22})+\Prob(B_{23}).
\end{equation}

The bounding of each of the probabilities $\mathrm{P}(B_{21})$, $\mathrm{P}(B_{22})$ , $\mathrm{P}(B_{23})$ is quite similar to the bounding of $\mathrm{P}(G\cap B_{1})$ -- because 

\begin{align*}
	\Prob(B_{21})&=\Prob(\sum_{i=1}^{n}Y_{i,21}\leq-n\delta_{21}), \\
	\Prob(B_{22})&=\Prob(\sum_{i=1}^{n}Y_{i,22}\geq n\delta_{22}) \textrm{ ,and } \\ 
	\Prob(B_{23})&= \Prob\!\left(\left|\sum_{i=1}^{n}Y_{i,23}\right|\geq n\delta_{23}\right).
\end{align*} 

Here $Y_{i,21}=U_i-\E U_1$, $Y_{i,22}=R_i^*-\E R_1^*$, and $Y_{i,23}=Z_i$ are centered, with thresholds $\delta_{21}=\E U_1/2$, $\delta_{22}=1$, and $\delta_{23}=(\E U_1)^2/[8(1+\E R_1^*)]$, respectively. Each has finite third absolute moment.

It follows that

\begin{equation}\label{GandB}
	\mathrm{P}(G\cap B)\leq \mathrm{P}(G\cap B_{1})+\mathrm{P}(B_{2})\leq\frac{\mathrm{C}}{n^{3/2}}, 
\end{equation}
where $\mathrm{C}$ depends on $\ell$, the measure $\mu$, and the choice of $\theta_{0}-$ but not on $n.$

On the other hand, if $\overline{R}\neq 0$ and $\overline{U}>0$, then $d_{-}=\displaystyle \frac{2\overline{Z}}{\overline{U}+\sqrt{\overline{U}^{2}-2\overline{Z}\overline{R}}}$. Here, the condition $\overline{U}>0$ is so the denominator of the latter ratio is nonzero. Thus, on the event $G\backslash B$ one has
\begin{equation}\label{inclusionrelation}
	\overline{U}>0 \textrm{ and } \displaystyle \hat{\theta}-\theta_{0}=\frac{2\overline{Z}}{\overline{U}+\sqrt{\overline{U}^{2}-2\overline{Z}\overline{R}}}\in[T_{-},\ T_{+}]
\end{equation}
where
\begin{equation}\label{Tpm}
\begin{aligned}
T_-&:=\frac{2\overline{Z}}{\overline{U}+\sqrt{\overline{U}^{2}+2\overline{Z}\,\overline{R^{*}}}},\\
T_+&:=\frac{2\overline{Z}}{\overline{U}+\sqrt{\overline{U}^{2}-2\overline{Z}\,\overline{R^{*}}}}.
\end{aligned}
\end{equation}
The root continuous at $\overline{R}=0$ is monotone nondecreasing in the remainder coefficient $r$. Indeed, on $G\backslash B$, if
$d(r)=2\overline{Z}/(\overline{U}+\sqrt{\overline{U}^{2}-2\overline{Z}r})$, then
\[
d'(r)=\frac{d(r)^2}{2\sqrt{\overline{U}^{2}-2\overline{Z}r}}\geq0
\qquad(-\overline{R^*}\leq r\leq\overline{R^*}).
\]
Thus the extremal values $r=\pm\overline{R^*}$ give the stated lower and upper bounds regardless of the sign of $\overline{Z}$.

\medskip

\subsection{General uniform and nonuniform bounds on the rate of convergence to normality for smooth nonlinear functions of sums of independent random vectors}

Denote the standard normal distribution function (d.f.) by $\Phi$. For any $\mathbb{R}^{d}$-valued random vector $\zeta$,

\[
\Vert\zeta\Vert_{p} :=(\mathrm{E}\Vert\zeta\Vert^{p})^{1/p} \textrm{ for any real } p\geq 1,
\]

where $\Vert$ . $\Vert$ denotes the Euclidean norm on $\mathbb{R}^{d}.$

Take any Borel-measurable functional $f:\mathbb{R}^{d}\rightarrow \mathbb{R}$ satisfying the following smoothness condition: there exist $\epsilon\in(0,\ \infty), M_{\epsilon}\in(0,\ \infty)$ , and a linear functional $L:\mathbb{R}^{d}\rightarrow \mathbb{R}$ such that

\begin{theorem}[Smoothness Condition]
	\begin{equation}\label{eqn:smoothness}
		|f(\displaystyle \mathrm{x})-L(\mathrm{x})|\leq\frac{M_{\epsilon}}{2}\Vert \mathrm{x}\Vert^{2} \textrm{ for all } \mathrm{x}\in \mathbb{R}^{d} \textrm{ with } \Vert \mathrm{x}\Vert\leq\epsilon .
	\end{equation}
	
\end{theorem}

Thus, $f(0)=0$ and $L$ necessarily coincides with the first Fr\'{e}chet derivative, $f'(0)$ , of the function $f$ at $0$. Moreover, for the smoothness condition to hold, it is enough that 

\[
M_{\epsilon}\geq M_{\epsilon}^{*} :=\displaystyle \sup\{\frac{1}{\Vert \mathrm{x}\Vert^{2}}|\frac{\mathrm{d}^{2}}{\mathrm{d}t^{2}}f(\mathrm{x}+t\mathrm{x})|_{t=0}|\ :\ \mathrm{x}\in \mathbb{R}^{d},\ 0<\Vert \mathrm{x}\Vert\leq\epsilon\}.
\]

Notice that $f$ does not need to be twice differentiable at $0$. One example is if $d=1$ and $f(x)= \displaystyle \frac{x}{1+|x|} \textrm{ for } x\in \mathbb{R}.$

Let $V, V_{1}$, . . . , $V_{n}$ be i.i.d. random vectors in $\mathbb{R}^{d}$, with $\E V=0$ and
$$
\overline{V}:=\frac{1}{n}\sum_{i=1}^{n}V_{i}.
$$
And let
\begin{equation}\label{sigmatilde}
	\tilde{\sigma} :=\Vert L(V)\Vert_{2}, v_{3} :=\Vert V\Vert_{3}, \textrm{ and } \varsigma_{3} :=\displaystyle \frac{\Vert L(V)\Vert_{3}}{\tilde{\sigma}}.
\end{equation}

\begin{theorem}\label{berryesseen}
	Suppose that the smoothness condition holds and that $\tilde{\sigma}>0$ and $v_{3}<\infty$. Then for all $z\in \mathbb{R}$
	\begin{equation}\label{berryesseen1}
		|\displaystyle \mathrm{P}(\frac{f(\overline{V})}{\tilde{\sigma}/\sqrt{n}}\leq z)-\Phi(z)|\leq\frac{\mathrm{C}}{\sqrt{n}},
	\end{equation}
	where $\mathrm{C}$ is a finite positive expression that depends only on the function $f$ and the moments $\tilde{\sigma}$, $\varsigma_{3}$, and $v_{3}$. Moreover, for any $\omega\in(0,\ \infty)$ and for all
	\begin{equation}\label{omega}
		z\in(0,\ \omega\sqrt{n}],
	\end{equation}
	
	one has
	
	\begin{equation}\label{berryesseen2}
		|\mathrm{P}(\frac{f(\overline{V})}{\tilde{\sigma}/\sqrt{n}}\leq z)-\Phi(z)|\leq\frac{\mathrm{C}_{\omega}}{z^{3}\sqrt{n}}
	\end{equation}
	
	where $C_{\omega}$ is a positive, finite, and only depends on $f$ through the smoothness condition, the moments $\tilde{\sigma}$, $\varsigma_{3}$, and $v_3$, and $\omega$.
\end{theorem}

\subsection{Applying bracketing}

Now let $d=3$ and then let

$$\mathcal{D} :=\{\mathrm{x}=(x_{1},\ x_{2},\ x_{3})\in \mathbb{R}^{d}=\mathbb{R}^{3}\ :\ x_{2}+\mathrm{E}U_{1}>0,\ (x_{2}+\mathrm{E}U_{1})^{2}>2|x_{1}||x_{3}+\mathrm{E}R_{1}^{*}|\}.$$

By \eqref{eqn:R} and assumptions 2 and 4 for $\ell$ , $\mathrm{E}U_{1}=I_2(\theta_{0})\in(0,\ \infty)$ and $\mathrm{E}R_{1}^{*}\in[0,\ \infty$). So, for some real $\epsilon>0$, the set $\mathcal{D}$ contains the $\epsilon$-neighborhood of the origin $0$ of $\mathbb{R}^{3}.$

Define functions $f_{\pm}:\mathbb{R}^{3}\rightarrow \mathbb{R}$ by the formula

\begin{equation}\label{definef}
\begin{aligned}
f_-(\mathrm{x})&=\frac{2x_1}{x_2+\E U_1+\sqrt{(x_2+\E U_1)^2+2x_1(x_3+\E R_1^*)}},\\
f_+(\mathrm{x})&=\frac{2x_1}{x_2+\E U_1+\sqrt{(x_2+\E U_1)^2-2x_1(x_3+\E R_1^*)}}.
\end{aligned}
\end{equation}

for $\mathrm{x}=(x_{1},\ x_{2},\ x_{3})\in \mathcal{D}$, and let $f_{\pm}(\mathrm{x}) :=0$ if $\mathrm{x}\in \mathbb{R}^{3}\backslash \mathcal{D}$. 

Clearly, $f_{\pm}(0)=0,$
\begin{equation}\label{Lplusminus}
	L_{\pm}(\displaystyle \mathrm{x}):=f_{\pm}'(0)(\mathrm{x})=\frac{x_{1}}{\mathrm{E}U_{1}}=\frac{x_{1}}{I_2(\theta_{0})}
\end{equation}
for $\mathrm{x}=(x_{1},\ x_{2},\ x_{3})\in \mathbb{R}^{3}$, and the smoothness condition \eqref{eqn:smoothness} holds for some $\epsilon$ and $M_{\epsilon}$ in $(0,\ \infty)$ --because, as was noted above, $\mathrm{E}U_{1}=I_2(\theta_{0})\in(0,\ \infty)$ and $\mathrm{E}R_{1}^{*}\in[0,\ \infty$), and hence the denominator of the ratio in \eqref{definef} is bounded away from $0$ for $\mathrm{x}=(x_{1},\ x_{2},\ x_{3})$ in a neighborhood of $0.$

Next, let
\begin{equation}\label{Vi}
	V_{i}:=(Z_{i},\ U_{i}-\mathrm{E}U_{i},\ R_{i}^{*}-\mathrm{E}R_{i}^{*})
\end{equation}
for $i=1,\ldots,n$, with $Z_{i}, U_{i}, R_{i}^{*}$ as defined in \eqref{eqn:R} and \eqref{eqn:Z} . Then, by \eqref{sigmatilde}, \eqref{Lplusminus} , and condition 2 , for $f=f_{\pm},$
\begin{equation}\label{sigmatilde1}
	\tilde{\sigma} = \sqrt{\frac{\mathrm{E}Z_{1}^{2}}{I_2(\theta_{0})^{2}}}=\frac{\sqrt{I_1(\theta_0)}}{I_2(\theta_{0})}>0
\end{equation}

and $ v_{3}^{3}=\mathrm{E}\Vert V\Vert^{3}<\infty$ by the third and fourth conditions. This shows that all the required conditions for Theorem~\ref{berryesseen} are satisfied for $f=f_{\pm}$.

Moreover, by \eqref{Vi}, \eqref{definef}, and \eqref{Tpm},
$$
T_{\pm}=f_{\pm}(\overline{V})
$$
on the event $G\backslash B$. So, by the inclusion relation in \eqref{inclusionrelation} (which holds on the event $G\backslash B=(G^{\mathrm{c}}\cup B)^{\mathrm{c}}$, where $\mathrm{c}$ denotes the complement) and \eqref{sigmatilde1} , inequality \eqref{berryesseen1} in Theorem~\ref{berryesseen} implies
$$
\mathrm{P}(\sqrt{n/I_1(\theta_{0})}I_2(\theta_0)(\hat{\theta}-\theta_{0})\leq z)\leq \mathrm{P}(\sqrt{n/I_1(\theta_{0})}I_2(\theta_0)f_{-}(\overline{V})\leq z)+\mathrm{P}(G^{\mathrm{c}}\cup B)
$$
$$
\leq\Phi(z)+\frac{\mathrm{C}}{\sqrt{n}}+\mathrm{P}(G^{\mathrm{c}}\cup B)
$$
and, quite similarly,
$$
\mathrm{P}(\sqrt{n/I_1(\theta_{0})}I_2(\theta_0)(\hat{\theta}-\theta_{0})\leq z)\geq \mathrm{P}(\sqrt{n/I_1(\theta_{0})}I_2(\theta_0)f_{+}(\overline{V})\leq z)-\mathrm{P}(G^{\mathrm{c}}\cup B)
$$
$$
\geq\Phi(z)-\frac{\mathrm{C}}{\sqrt{n}}-\mathrm{P}(G^{\mathrm{c}}\cup B)\ ,
$$
for all real $z$. Note that $\mathrm{P}(G^{\mathrm{c}}\cup B)=\mathrm{P}(G^{\mathrm{c}})+\mathrm{P}(G\cap B)$ . It follows now by \eqref{eqn:goodevent} and \eqref{GandB} that
\begin{equation}
	|\displaystyle \mathrm{P}(\sqrt{n/I_1(\theta_{0})}I_2(\theta_0)(\hat{\theta}-\theta_{0})\leq z)-\Phi(z)|\leq\frac{\mathrm{C}}{\sqrt{n}}+\frac{\mathrm{C}_{B}}{n^{3/2}}+\mathrm{P}(|\hat{\theta}-\theta_{0}|>\delta)
\end{equation}
for all real $z$. To retain the nonuniform rate, we strengthen the bad-event bound before applying \eqref{berryesseen2}.

\begin{lemma}[Fixed-threshold deviations]\label{lem:fixed-threshold}
If $Y_i$ are i.i.d., $\E Y_1=0$ and $\E|Y_1|^3<\infty$, then for every fixed $a>0$,
\[
\Prob\!\left(\left|\sum_{i=1}^nY_i\right|\geq na\right)\leq C_a n^{-2}.
\]
Consequently, $\Prob(G\cap B)\leq C'n^{-2}$.
\end{lemma}
\begin{proof}
If $s^2:=\E Y_1^2=0$, the probability is zero. Otherwise apply Theorem~\ref{berryesseen}, with the linear map $f(y)=y$, to $Y_i$ and to $-Y_i$ at $t=a\sqrt{n}/(2s)$, taking $\omega=a/(2s)$. The strict threshold $na/2$ bounds the event with the weak threshold $na$, so atoms at the endpoints cause no difficulty. The union bound yields
\[
\Prob\!\left(\left|\sum_{i=1}^nY_i\right|\geq na\right)
\leq2\{1-\Phi(t)\}+\frac{2C_{\omega}}{t^3\sqrt{n}}
\leq C_a n^{-2},
\]
since the Gaussian tail decreases exponentially in $n$. Apply this result to $Y_i=(U_i-\delta R_i^*)-\E(U_i-\delta R_i^*)$ and to $Y_{i,21}$, $Y_{i,22}$, $Y_{i,23}$ with their positive fixed thresholds. The containment in \eqref{eqn:bound_intersection} and the union bound for $B_2$ give the conclusion.
\end{proof}

Using this lemma and \eqref{berryesseen2} in the same bracketing argument gives
\begin{equation}
	|\displaystyle \mathrm{P}(\sqrt{n/I_1(\theta_{0})}I_2(\theta_0)(\hat{\theta}-\theta_{0})\leq z)-\Phi(z)|\leq\frac{\mathrm{C}_{\omega}}{z^{3}\sqrt{n}}+\frac{\mathrm{C}_{B}}{n^{2}}+\mathrm{P}(|\hat{\theta}-\theta_{0}|>\delta)
\end{equation}
for $z$ as in \eqref{omega}. In this range,
\[
\frac{1}{n^2}\leq\frac{\omega^3}{z^3\sqrt{n}},
\]
so the additional bad-event term is absorbed into the nonuniform constant. The weaker $n^{-3/2}$ bound alone would not justify this step at $z=\omega\sqrt n$.

Given rather standard regularity conditions, the remainder term $\mathrm{P}(|\hat{\theta}-\theta_{0}|>\delta)$ typically decreases exponentially fast in $n$ and thus is negligible compared with the error term $\displaystyle \frac{\mathrm{C}}{\sqrt{n}}$, and even with the error term $\displaystyle \frac{\mathrm{C}}{z^{3}\sqrt{n}}$ under condition \eqref{omega}. Some details on this can be found in the following section.

\subsection{Bounding the remainder}
\label{sec:remainder}

We bound the remainder using conditions (A5) and (A6). On the event
$$
E_n:=\set*{\sup_{\theta\in\Theta}\abs{n^{-1}L_n(\theta)-\E\,\ell_{X,Z}(\theta)}<D(\delta)/2},
$$
every $\theta\in\Theta$ with $\abs{\theta-\theta_0}\ge\delta$ satisfies
$$
n^{-1}L_n(\theta)\ <\ \E\,\ell_{X,Z}(\theta)+\tfrac{1}{2}D(\delta)
\ \le\ \E\,\ell_{X,Z}(\theta_0)-\tfrac{1}{2}D(\delta)
\ <\ n^{-1}L_n(\theta_0),
$$
where the middle inequality is the definition of $D(\delta)$ in (A5) and the outer two hold on $E_n$. No such $\theta$ can therefore maximize $L_n$, so $E_n\subseteq\set{\abs{\hat\theta-\theta_0}<\delta}$. By (A6),
\begin{equation}\label{eq:exp-decay}
	\mathrm{P}(\abs{\hat{\theta}-\theta_{0}}>\delta)\ \leq\ \Prob(E_n^{\mathrm{c}})\ \leq\ C_1 e^{-c_2 n}.
\end{equation}
Since $n^{-3/2}=o(n^{-1/2})$, the extra term in the uniform bound is absorbed into its constant. The localization probability decays exponentially in $n$ and is thus negligible beside both $\mathrm{C}/\sqrt{n}$ and $\mathrm{C}_\omega/(z^{3}\sqrt{n})$ under condition \eqref{omega}.

\begin{remark}[Why not concavity]
\label{rem:concavity}
Pinelis \cite{Pinelis2017} bounds the corresponding remainder for maximum likelihood estimators by assuming the per-observation criterion is concave in $\theta$, which is natural for log-likelihoods of exponential families. That route is unavailable in our setting. Whenever $\sup_{\theta}p_{\theta}(x)<\infty$, the criterion $\ell_{x,z}(\theta)=-(p_{\theta}(x)-z)^{2}$ is bounded on $\Theta$ and non-constant; since a concave function that is bounded below on $\mathbb{R}$ is constant, $\ell_{x,z}$ is concave for no such family---in particular for neither of the families used in Section~\ref{sec:elicitation-experiments}. Conditions (A5) and (A6) replace it, and this is the adjustment to the assumptions of \cite{Pinelis2017} that the M-estimation setting requires.

Condition (A5) is mild. Because $\E[Z\mid x]=p_{\theta_{0}}(x)$, the cross term vanishes and
$$
\E\,\ell_{X,Z}(\theta_{0})-\E\,\ell_{X,Z}(\theta)=\E_{x\sim p_{\theta_{0}}}\bkt*{\paren*{p_{\theta}(x)-p_{\theta_{0}}(x)}^{2}},
$$
which does not depend on $\sigma$ and vanishes only at $\theta_{0}$ for an identifiable family. Expanding, $D(\theta)=\E[p_{\theta}^{2}]-2\,\E[p_{\theta}p_{\theta_{0}}]+\E[p_{\theta_{0}}^{2}]$ with all expectations under $p_{\theta_{0}}$. The cross term vanishes as $p_{\theta}$ separates from $p_{\theta_{0}}$, so $D(\theta)\to\E[p_{\theta_{0}}^{2}]+\lim\E[p_{\theta}^{2}]\ \ge\ \E[p_{\theta_{0}}^{2}]>0$, and (A5) therefore holds on an unbounded $\Theta$ with no compactness assumption. The two limits differ between our families: for the normal location family $\E[p_{\theta}^{2}]\to0$, so $D(\theta)\to\E[p_{\theta_{0}}^{2}]=1/(2\pi\sqrt{3})\approx0.0919$; for $\mathrm{Beta}(\theta,2)$ with $\theta_{0}=3$ the density concentrates near $x=1$ rather than escaping, and $\E_{p_{\theta_{0}}}[p_{\theta}^{2}]=12\theta^{2}(\theta+1)^{2}B(2\theta+1,4)\to9/2$, so $D(\theta)\to9/2+\E[p_{\theta_{0}}^{2}]\approx6.56$. In both cases the limit is bounded away from zero, which is all (A5) requires. Condition (A6) is a uniform law of large numbers with an exponential rate; it is discharged in Section~\ref{sec:appendix-a6}, where it is shown to follow from two elementary conditions that hold for both families used here---one of which is precisely the boundedness that rules concavity out---and, without any compactness assumption, from a third condition of tightness at infinity (Proposition~\ref{prop:a6-unbounded}). The limit $\lim_{\abs{\theta}\to\infty}\E[p_\theta^2]$ just computed is exactly what decides that third condition: it is $0$ for the normal location family and $9/2$ for the beta shape family, and Remark~\ref{rem:a6-beta-unbounded} shows the latter makes (A6) false on $[1.1,\infty)$. These are the same two conditions under which consistency was obtained in Section~\ref{sec:consistency}, where the parallel obstruction to a monotonicity argument is noted.
\end{remark}

\section{Appendix: verification of condition (A6)}
\label{sec:appendix-a6}

Condition (A6) is a uniform law of large numbers with an exponential rate. We show it follows from the following two conditions.

\begin{enumerate}
	\item[(B1)] $\Theta\subset\mathbb{R}$ is compact, with diameter $T:=\sup\Theta-\inf\Theta$, and both
	$$
	\bar P:=\sup_{\theta\in\Theta}\ \sup_{x\in\mathcal{X}}\ p_{\theta}(x)<\infty
	\qquad\text{and}\qquad
	\dot P:=\sup_{\theta\in\Theta}\ \sup_{x\in\mathcal{X}}\ \abs*{\partial_{\theta}p_{\theta}(x)}<\infty .
	$$
	\item[(B2)] $Z=p_{\theta_{0}}(X)(1+\sigma\xi)$ with $\xi$ independent of $X$, $\E\xi=0$, $\E\xi^{2}=1$, and $\xi$ sub-Gaussian: $\E e^{\lambda\xi}\le e^{\lambda^{2}b^{2}/2}$ for all real $\lambda$ and some $b<\infty$.
\end{enumerate}

\begin{lemma}
\label{lem:a6}
Under (B1) and (B2), for every $\eta>0$ there are constants $C_1,c_2\in(0,\infty)$, depending on $\eta$ but not on $n$, such that
$$
\Prob\Big(\sup_{\theta\in\Theta}\abs*{n^{-1}L_n(\theta)-\E\,\ell_{X,Z}(\theta)}\ \ge\ \eta\Big)\ \le\ C_1 e^{-c_2 n}.
$$
In particular (A6) holds, on taking $\eta=D(\delta)/2$.
\end{lemma}

\begin{proof}
By (B2), $\abs{Z}\le\bar P(1+\sigma\abs{\xi})$, so
\begin{equation}\label{eq:envelope}
	\sup_{\theta\in\Theta}\abs*{\ell_{X,Z}(\theta)}=\sup_{\theta\in\Theta}\paren*{p_{\theta}(X)-Z}^{2}\le\paren*{\bar P+\abs{Z}}^{2}\le\bar P^{2}\paren*{2+\sigma\abs{\xi}}^{2}=:M,
\end{equation}
and, since $\partial_{\theta}\ell_{x,z}(\theta)=-2(p_{\theta}(x)-z)\,\partial_{\theta}p_{\theta}(x)$,
\begin{equation}\label{eq:lipschitz}
	\sup_{\theta\in\Theta}\abs*{\partial_{\theta}\ell_{X,Z}(\theta)}\ \le\ 2\dot P\paren*{\bar P+\abs{Z}}\ \le\ 2\dot P\bar P\paren*{2+\sigma\abs{\xi}}=:\Lambda .
\end{equation}
Thus $\theta\mapsto\ell_{X,Z}(\theta)$ is Lipschitz on $\Theta$ with the random constant $\Lambda$. As $\xi$ is sub-Gaussian, $\Lambda$ is sub-Gaussian and $M$, being a squared sub-Gaussian variable, is sub-exponential; both have finite means, and we write $\bar\Lambda:=\E\Lambda\in(0,\infty)$.

Fix $\eta>0$ and put $\delta:=\eta/(6\bar\Lambda)$. By compactness choose $\theta_{1},\ldots,\theta_{N}\in\Theta$ with $N\le\lceil T/\delta\rceil+1$ such that every $\theta\in\Theta$ lies within $\delta$ of some $\theta_{j}$; note $N$ depends on $\eta$ but not on $n$. Write $\bar G_n(\theta):=n^{-1}L_n(\theta)-\E\,\ell_{X,Z}(\theta)$. If $\abs{\theta-\theta_{j}}\le\delta$ then, by \eqref{eq:lipschitz} applied to each summand and to the expectation,
$$
\abs*{\bar G_n(\theta)-\bar G_n(\theta_{j})}\ \le\ \delta\paren*{n^{-1}\textstyle\sum_{i=1}^{n}\Lambda_{i}+\bar\Lambda},
$$
so that
$$
\sup_{\theta\in\Theta}\abs*{\bar G_n(\theta)}\ \le\ \max_{j\le N}\abs*{\bar G_n(\theta_{j})}+\delta\paren*{n^{-1}\textstyle\sum_{i=1}^{n}\Lambda_{i}+\bar\Lambda}.
$$
On the event $A_n:=\set{n^{-1}\sum_{i}\Lambda_{i}\le 2\bar\Lambda}$ the second term is at most $3\delta\bar\Lambda=\eta/2$. Hence
$$
\Prob\Big(\sup_{\theta}\abs*{\bar G_n(\theta)}\ge\eta\Big)\ \le\ \Prob(A_n^{\mathrm{c}})+\sum_{j\le N}\Prob\big(\abs*{\bar G_n(\theta_{j})}\ge\eta/2\big).
$$
Both terms decay exponentially. The variables $\Lambda_{i}$ are i.i.d.\ and sub-exponential, so a standard Bernstein inequality for sub-exponential summands \cite{Vershynin2018} gives $\Prob(A_n^{\mathrm{c}})=\Prob\big(n^{-1}\sum_{i}\Lambda_{i}-\bar\Lambda\ge\bar\Lambda\big)\le e^{-c_{3}n}$ with $c_{3}>0$ depending only on the sub-exponential parameters of $\Lambda$. For each fixed $\theta_{j}$ the summands $\ell_{X_i,Z_i}(\theta_{j})$ are i.i.d.\ and, by \eqref{eq:envelope}, dominated by the sub-exponential envelope $M$; the same inequality gives $\Prob(\abs{\bar G_n(\theta_{j})}\ge\eta/2)\le 2e^{-c_{4}n}$ with $c_{4}>0$ depending on $\eta$ and on the parameters of $M$, but not on $n$ or $j$. Therefore
$$
\Prob\Big(\sup_{\theta}\abs*{\bar G_n(\theta)}\ge\eta\Big)\ \le\ e^{-c_{3}n}+2Ne^{-c_{4}n}\ \le\ (1+2N)e^{-\min(c_{3},c_{4})n},
$$
which is the assertion with $C_1=1+2N$ and $c_2=\min(c_{3},c_{4})$.
\end{proof}

Both conditions hold for the families used in Section~\ref{sec:elicitation-experiments}. For the normal location family on $\Theta=[-3,7]$, the interval searched in the experiments, $\bar P=(2\pi)^{-1/2}\approx0.3989$ and $\dot P=\sup_{u}\abs{\varphi(u)u}=\varphi(1)\approx0.2420$, with $T=10$; and $\xi\sim\mathcal{N}(0,1)$ satisfies (B2) with $b=1$. The beta shape family needs more care. Here $p_{\theta}(x)=\theta(\theta+1)x^{\theta-1}(1-x)$, which is unbounded on $(0,1)$ when $\theta<1$, and
$$
\partial_{\theta}p_{\theta}(x)=(2\theta+1)x^{\theta-1}(1-x)+\theta(\theta+1)(1-x)\,x^{\theta-1}\log x ,
$$
whose second term reduces to $2(1-x)\log x$ at $\theta=1$ and is therefore unbounded as $x\to0$. Condition (B1) thus fails at $\theta=1$ as well as below it, and requires $\Theta$ to be bounded away from $1$ from above: on $\Theta=[1.1,10]$, which contains $\theta_{0}=3$ with room to spare, $\bar P\approx4.262$ and $\dot P\approx7.399$, with $T=8.9$. The experiments of Section~\ref{sec:elicitation-experiments} search exactly this interval, so the guarantee applies to the estimator as implemented.


\subsection*{Removing the compactness assumption}

Compactness enters the proof of Lemma~\ref{lem:a6} at a single point: the $\delta$-net of $\Theta$ has $N\le\lceil T/\delta\rceil+1$ points, and the union bound multiplies the pointwise Bernstein bound by $N$. On an unbounded $\Theta$ the net is infinite and the bound is vacuous. Nothing else in the proof uses boundedness: the envelope~\eqref{eq:envelope} and the Lipschitz constant~\eqref{eq:lipschitz} are uniform in $\theta$ as soon as $\bar P$ and $\dot P$ are finite. The remedy is to cover a bounded central piece of $\Theta$ by a net as before and to control the two tails by a single integrable envelope, which requires no net at all. The price is one further condition on the family.

\begin{enumerate}
	\item[(B1$'$)] $\Theta\subset\mathbb{R}$, not necessarily bounded, and $\bar P<\infty$, $\dot P<\infty$ as in (B1), the suprema now taken over all of $\Theta$.
	\item[(B3)] \emph{(Vanishing at infinity.)} With
	$$
	u_R(x):=\sup_{\theta\in\Theta:\ \abs{\theta-\theta_0}\ge R}\ p_\theta(x),
	$$
	one has $\E_{X\sim p_{\theta_0}}\,u_R(X)\to0$ as $R\to\infty$.
\end{enumerate}

The supremum defining $u_R$ is over an uncountable set, but $\theta\mapsto p_\theta(x)$ is continuous for each fixed $x$ in the families considered here, so the supremum is unchanged if $\theta$ is restricted to a countable dense subset of $\set{\theta\in\Theta:\abs{\theta-\theta_0}\ge R}$; hence $u_R$ is measurable, and so are $w_R$ of~\eqref{eq:tail-envelope} and $v_R$ of (B4) below.

Since $0\le u_R\le\bar P$, condition (B3) also gives $\E\,u_R(X)^2\le\bar P\,\E\,u_R(X)\to0$.

\begin{prop}[(A6) without compactness]
\label{prop:a6-unbounded}
Under (B1$'$), (B2) and (B3), for every $\eta>0$ there are constants $C_1,c_2\in(0,\infty)$, depending on $\eta$ but not on $n$, such that
$$
\Prob\Big(\sup_{\theta\in\Theta}\abs*{n^{-1}L_n(\theta)-\E\,\ell_{X,Z}(\theta)}\ \ge\ \eta\Big)\ \le\ C_1 e^{-c_2 n}.
$$
In particular (A6) holds on $\Theta$, on taking $\eta=D(\delta)/2$.
\end{prop}

\begin{proof}
Write $\bar G_n(\theta):=n^{-1}L_n(\theta)-\E\,\ell_{X,Z}(\theta)$ as in the proof of Lemma~\ref{lem:a6}. Expanding the square,
$$
\ell_{x,z}(\theta)=-(p_\theta(x)-z)^2=-z^2-q_{x,z}(\theta),\qquad q_{x,z}(\theta):=p_\theta(x)^2-2z\,p_\theta(x),
$$
and for $\abs{\theta-\theta_0}\ge R$ the second term satisfies $\abs{q_{x,z}(\theta)}\le w_R(x,z)$ with
\begin{equation}\label{eq:tail-envelope}
	w_R(x,z):=u_R(x)^2+2\abs{z}\,u_R(x),
\end{equation}
which does not depend on $\theta$. Writing $C_n:=n^{-1}\sum_{i}Z_i^2-\E Z^2$, it follows that for every $\theta$ with $\abs{\theta-\theta_0}\ge R$,
\begin{equation}\label{eq:tail-bound}
	\abs*{\bar G_n(\theta)}\ \le\ \abs{C_n}\ +\ n^{-1}\textstyle\sum_{i}w_R(X_i,Z_i)\ +\ \E\,w_R(X,Z),
\end{equation}
and the right-hand side does not depend on $\theta$ either. Two facts about $w_R$ are needed. First, using $u_R\le\bar P$, the representation $Z=p_{\theta_0}(X)(1+\sigma\xi)$ of (B2), and the independence of $\xi$ and $X$,
$$
\E\,w_R(X,Z)\ \le\ \bar P\,\E u_R(X)+2\,\E\abs{1+\sigma\xi}\ \E\big[p_{\theta_0}(X)\,u_R(X)\big]\ \le\ \bar P\big(3+2\sigma\E\abs{\xi}\big)\,\E\,u_R(X),
$$
which tends to zero as $R\to\infty$ by (B3). Second, $0\le w_R(X,Z)\le\bar P^2\big(3+2\sigma\abs{\xi}\big)$, so $w_R(X,Z)$ is sub-Gaussian with parameters depending only on $\bar P$, $\sigma$ and $b$, and not on $R$.

Fix $\eta>0$ and choose $R=R(\eta)$ so large that $\E\,w_R(X,Z)\le\eta/8$. Split $\Theta=\Theta_{\mathrm{c}}\cup\Theta_{\mathrm{t}}$ with $\Theta_{\mathrm{c}}:=\Theta\cap[\theta_0-R,\theta_0+R]$ and $\Theta_{\mathrm{t}}:=\Theta\setminus\Theta_{\mathrm{c}}$, so that
$$
\Prob\Big(\sup_{\Theta}\abs{\bar G_n}\ge\eta\Big)\ \le\ \Prob\Big(\sup_{\Theta_{\mathrm{c}}}\abs{\bar G_n}\ge\eta\Big)+\Prob\Big(\sup_{\Theta_{\mathrm{t}}}\abs{\bar G_n}\ge\eta\Big).
$$
The set $\Theta_{\mathrm{c}}$ is bounded with diameter at most $2R$, and the proof of Lemma~\ref{lem:a6} used compactness only to produce a finite $\delta$-net; it therefore applies verbatim to $\Theta_{\mathrm{c}}$ and bounds the first term by $(1+2N)e^{-c_2 n}$ with $N\le\lceil 2R/\delta\rceil+1$ and $\delta=\eta/(6\bar\Lambda)$. For the second term, by~\eqref{eq:tail-bound} and the choice of $R$,
$$
\Big\{\sup_{\Theta_{\mathrm{t}}}\abs{\bar G_n}\ge\eta\Big\}\ \subseteq\ \Big\{\abs{C_n}\ge\tfrac{\eta}{4}\Big\}\ \cup\ \Big\{n^{-1}\textstyle\sum_{i}w_R(X_i,Z_i)-\E\,w_R(X,Z)\ \ge\ \tfrac{\eta}{4}\Big\},
$$
since on the complement of the right-hand side, \eqref{eq:tail-bound} gives $\abs{\bar G_n(\theta)}<\tfrac{\eta}{4}+\big(\tfrac{\eta}{4}+\E w_R\big)+\E w_R\le\tfrac{3\eta}{4}$ for every $\theta\in\Theta_{\mathrm{t}}$, so that $\sup_{\Theta_{\mathrm{t}}}\abs{\bar G_n}\le\tfrac{3\eta}{4}<\eta$. The $Z_i^2$ are i.i.d.\ and sub-exponential, so Bernstein's inequality bounds the first event by $2e^{-c_5 n}$; the $w_R(X_i,Z_i)$ are i.i.d.\ and sub-Gaussian, so Hoeffding's inequality in its sub-Gaussian form \cite{Vershynin2018} bounds the second by $e^{-c_6 n}$, with $c_5,c_6>0$ depending on $\eta$ and on the parameters of the envelopes but not on $n$. Altogether
$$
\Prob\Big(\sup_{\Theta}\abs{\bar G_n}\ge\eta\Big)\ \le\ (1+2N)e^{-c_2 n}+2e^{-c_5 n}+e^{-c_6 n}\ \le\ (4+2N)\,e^{-\min(c_2,c_5,c_6)\,n},
$$
which is the claim.
\end{proof}

\begin{corollary}[Normal location family on the whole line]
\label{cor:a6-normal}
Let $p_\theta(x)=\varphi(x-\theta)$ with $\varphi$ the standard normal density, $\Theta=\mathbb{R}$, and let (B2) hold. Then (A6) holds on $\Theta=\mathbb{R}$. Explicitly, $u_R(x)=\varphi\big((R-\abs{x-\theta_0})_+\big)$ and
$$
\E_{\theta_0}\,u_R(X)\ \le\ 2\varphi(0)\,\bar\Phi(R/2)+\varphi(R/2),
$$
where $\bar\Phi:=1-\Phi$, so (B3) holds at a Gaussian rate.
\end{corollary}

\begin{proof}
(B1$'$) holds with $\bar P=\varphi(0)$ and $\dot P=\varphi(1)$, the same values as on any bounded interval, since both suprema are attained at finite $\theta$. For $u_R$: the supremum of $\varphi(x-\theta)$ over $\abs{\theta-\theta_0}\ge R$ is attained at the admissible $\theta$ nearest to $x$, which is $x$ itself when $\abs{x-\theta_0}\ge R$ and $\theta_0\pm R$ otherwise, giving the stated formula. With $U:=X-\theta_0\sim\mathcal{N}(0,1)$, split on $\abs{U}\le R/2$: there $R-\abs{U}\ge R/2$ and $\varphi$ is decreasing on $[0,\infty)$, so $u_R(X)\le\varphi(R/2)$; elsewhere $u_R\le\varphi(0)$. Hence $\E u_R(X)\le\varphi(R/2)+\varphi(0)\Prob(\abs{U}>R/2)$, which is the bound, and it tends to zero as $R\to\infty$.
\end{proof}

For the noise level $\sigma=0.3$ and the separation $\delta=0.25$, the threshold $\eta=D(\delta)/2\approx9.5\times10^{-4}$ requires $\E u_R(X)\le1.7\times10^{-4}$, which the bound of Corollary~\ref{cor:a6-normal} delivers at $R=8$: the netted piece is $[\theta_0-8,\theta_0+8]=[-6,10]$, and the two tails contribute only the two $\theta$-free events of the proof. The interval $[-3,7]$ searched in Section~\ref{sec:elicitation-experiments} is thus not a restriction imposed by the theory but a convenience of the implementation.

\begin{remark}[The beta shape family, and why compactness cannot be removed there]
\label{rem:a6-beta-unbounded}
For $p_\theta(x)=\theta(\theta+1)x^{\theta-1}(1-x)$ on $\Theta=[1.1,\infty)$ the conclusion of Proposition~\ref{prop:a6-unbounded} is false, and the reason is visible in the computation of Remark~\ref{rem:concavity}. For every fixed sample, $p_\theta(X_i)\to0$ as $\theta\to\infty$ for each $i$, since $X_i<1$ almost surely; hence $n^{-1}L_n(\theta)\to-n^{-1}\sum_iZ_i^2$. On the population side, however, $\E_{\theta_0}p_\theta(X)^2\to9/2$ while $\E_{\theta_0}[p_\theta(X)p_{\theta_0}(X)]\to0$, so $\E\,\ell_{X,Z}(\theta)\to-\big(9/2+\E Z^2\big)$. Therefore
$$
\lim_{\theta\to\infty}\bar G_n(\theta)\ =\ \tfrac92-C_n\qquad\text{almost surely},
$$
and $\sup_{\Theta}\abs{\bar G_n}\ge\tfrac92-\abs{C_n}$, so (A6) fails with probability tending to one for every $\eta<9/2$. In the language of the proposition, (B1$'$) fails because the mode of $p_\theta$, at $x^*=(\theta-1)/\theta$, has height $(\theta+1)(1-1/\theta)^{\theta-1}\sim\theta/e$, so $\bar P=\infty$; and while $\E u_R(X)<\infty$, one has $u_R(x)\sim4e^{-2}/(1-x)$ as $x\to1$, so $\E u_R(X)^2=\infty$ for every $R$. (A second, random contribution arises from the largest observation: with $\theta_n$ chosen to maximise $p_\theta(X_{(n)})$, the term $n^{-1}p_{\theta_n}(X_{(n)})^2$ converges in law to $16e^{-4}/W^2$, where $W$ is the nondegenerate limit of $n^{1/2}\min_i(1-X_i)$; this is what makes the empirical supremum occasionally very large, but the deterministic limit above already suffices to refute (A6).) More generally, whenever $p_\theta(x)\to0$ pointwise as $\abs{\theta}\to\infty$, condition (A6) on an unbounded $\Theta$ forces $\lim_{\abs{\theta}\to\infty}\E_{\theta_0}p_\theta(X)^2=0$, i.e.\ $\lim D(\theta)=\E p_{\theta_0}^2$; by Remark~\ref{rem:concavity} the normal location family satisfies this and the beta shape family does not.

Two qualifications. First, on the bounded interval $[1.1,10]$ used in the experiments, Lemma~\ref{lem:a6} applies and nothing changes. Second, the failure of (A6) on $[1.1,\infty)$ does not refute consistency there. The argument of Section~\ref{sec:remainder} uses the deviation at $\theta\ne\theta_0$ in one direction only, and the term responsible for the failure, $-n^{-1}\sum_ip_\theta(X_i)^2$, enters $L_n$ with a sign that only helps: it can be discarded in a one-sided bound. Proposition~\ref{prop:consistency-unbounded} below carries this out and gives exponential consistency on $[1.1,\infty)$, where (A6) is false.
\end{remark}

\subsection*{Consistency without (A6): a one-sided argument}

Everything after Section~\ref{sec:remainder} uses (A5) and (A6) only through the conclusion~\eqref{eq:exp-decay}. The following proposition reaches~\eqref{eq:exp-decay} by a direct comparison of $L_n$ on the tails with $L_n(\theta_0)$, under a condition weaker than (B3) that the beta shape family satisfies.

\begin{enumerate}
	\item[(B4)] \emph{(Cross-term tightness.)} With $u_R$ as in (B3) and $v_R(x):=p_{\theta_0}(x)\,u_R(x)$, one has $V:=\sup_{x}v_R(x)<\infty$ for some $R$, and $\E_{X\sim p_{\theta_0}}v_R(X)\to0$ as $R\to\infty$.
\end{enumerate}

Since $v_R\le\bar P\,u_R$, condition (B4) follows from (B1$'$) and (B3); its point is that it can hold when $\bar P=\infty$.

\begin{prop}[Exponential consistency on an unbounded parameter set]
\label{prop:consistency-unbounded}
Let $\Theta\subset\mathbb{R}$ contain $\theta_0$. Assume (B2), that (A5) holds on $\Theta$, that (B1) holds on every bounded subset of $\Theta$, and (B4). Then for every $\delta>0$ there are constants $C_1,c_2\in(0,\infty)$, not depending on $n$, such that
$$
\Prob\big(\abs{\hat\theta-\theta_0}\ge\delta\big)\ \le\ C_1e^{-c_2n}.
$$
\end{prop}

\begin{proof}
Put $\kappa:=\E_{\theta_0}p_{\theta_0}(X)^2>0$. Under (B2), $p_{\theta_0}(X_i)-Z_i=-\sigma\xi_ip_{\theta_0}(X_i)$, so the sample criterion at the truth is
\begin{equation}\label{eq:Ln-at-truth}
	n^{-1}L_n(\theta_0)\ =\ -\sigma^2\,n^{-1}\textstyle\sum_i\xi_i^2\,p_{\theta_0}(X_i)^2 .
\end{equation}
For any $\theta$, expanding the square and discarding the nonpositive term $-n^{-1}\sum_ip_\theta(X_i)^2$,
$$
\begin{aligned}
n^{-1}L_n(\theta)&=-n^{-1}\textstyle\sum_ip_\theta(X_i)^2+2\,n^{-1}\sum_iZ_ip_\theta(X_i)-n^{-1}\sum_iZ_i^2\\
&\leq2\,n^{-1}\textstyle\sum_i\abs{Z_i}\,p_\theta(X_i)-n^{-1}\sum_iZ_i^2.
\end{aligned}
$$
Fix $R$ and let $\Theta_{\mathrm{t}}:=\set{\theta\in\Theta:\abs{\theta-\theta_0}\ge R}$. For $\theta\in\Theta_{\mathrm{t}}$ we have $p_\theta(X_i)\le u_R(X_i)$, and $\abs{Z_i}u_R(X_i)=v_R(X_i)\abs{1+\sigma\xi_i}$ by (B2), whence
\begin{equation}\label{eq:tail-onesided}
	\sup_{\theta\in\Theta_{\mathrm{t}}}n^{-1}L_n(\theta)\ \le\ 2\bar v_n-n^{-1}\textstyle\sum_iZ_i^2,\qquad \bar v_n:=n^{-1}\textstyle\sum_iv_R(X_i)\abs{1+\sigma\xi_i}.
\end{equation}
Subtracting~\eqref{eq:Ln-at-truth} from~\eqref{eq:tail-onesided} and using $Z_i^2-\sigma^2\xi_i^2p_{\theta_0}(X_i)^2=p_{\theta_0}(X_i)^2(1+2\sigma\xi_i)$,
$$
\sup_{\theta\in\Theta_{\mathrm{t}}}n^{-1}L_n(\theta)\ <\ n^{-1}L_n(\theta_0)
\qquad\text{whenever}\qquad
2\bar v_n\ <\ m_n:=n^{-1}\textstyle\sum_ip_{\theta_0}(X_i)^2\,(1+2\sigma\xi_i).
$$
Now $\E m_n=\kappa$ since $\E\xi=0$ and $\xi$ is independent of $X$, and $\E\bar v_n=\E\abs{1+\sigma\xi}\,\E v_R(X)$. By (B4) choose $R\ge\delta$ so large that $\E\abs{1+\sigma\xi}\,\E v_R(X)\le\kappa/8$. The summands of $m_n$ are bounded by $\bar P_{0}^2(1+2\sigma\abs{\xi_i})$ with $\bar P_0:=\sup_xp_{\theta_0}(x)$, and those of $\bar v_n$ by $V(1+\sigma\abs{\xi_i})$; both are sub-Gaussian by (B2), so Hoeffding's inequality in its sub-Gaussian form \cite{Vershynin2018} gives
$$
\Prob\big(m_n\le\kappa/2\big)\le e^{-c_7n},\qquad \Prob\big(\bar v_n\ge\kappa/4\big)\le\Prob\big(\bar v_n-\E\bar v_n\ge\kappa/8\big)\le e^{-c_8n},
$$
with $c_7,c_8>0$ depending on $\kappa$, $\bar P_0$, $V$, $\sigma$ and $b$ but not on $n$. Off these two events, $2\bar v_n<\kappa/2<m_n$, so every $\theta\in\Theta_{\mathrm{t}}$ has $n^{-1}L_n(\theta)<n^{-1}L_n(\theta_0)$.

It remains to treat $\Theta_{\mathrm{c}}:=\Theta\cap(\theta_0-R,\theta_0+R)$, a bounded set on which (B1) holds by hypothesis. Lemma~\ref{lem:a6} applies to $\Theta_{\mathrm{c}}$ (its proof used compactness only to produce a finite net, and $\Theta_{\mathrm{c}}$ has diameter at most $2R$), and gives $\Prob\big(\sup_{\Theta_{\mathrm{c}}}\abs{\bar G_n}\ge D(\delta)/2\big)\le C_1'e^{-c_2'n}$, where $D(\delta)>0$ is the separation of (A5). On the complement, the argument of Section~\ref{sec:remainder} shows that every $\theta\in\Theta_{\mathrm{c}}$ with $\abs{\theta-\theta_0}\ge\delta$ has $n^{-1}L_n(\theta)<n^{-1}L_n(\theta_0)$. Combining, off an event of probability at most $C_1'e^{-c_2'n}+e^{-c_7n}+e^{-c_8n}$, no $\theta\in\Theta$ with $\abs{\theta-\theta_0}\ge\delta$ maximises $L_n$, which is the claim with $C_1=C_1'+2$ and $c_2=\min(c_2',c_7,c_8)$.
\end{proof}

\begin{corollary}[Beta shape family on $[1.1,\infty)$]
\label{cor:consistency-beta}
For $p_\theta(x)=\theta(\theta+1)x^{\theta-1}(1-x)$ with $\theta_0=3$ and $\Theta=[1.1,\infty)$, under (B2), $\Prob(\abs{\hat\theta-\theta_0}\ge\delta)\le C_1e^{-c_2n}$ for every $\delta>0$.
\end{corollary}

\begin{proof}
(A5) on $[1.1,\infty)$ and (B1) on every $[1.1,B]$ were verified above. For (B4): $p_{\theta_0}(x)p_\theta(x)=12\,\theta(\theta+1)\,x^{\theta+1}(1-x)^2$ is maximised over $x$ at $x=(\theta+1)/(\theta+3)$, where it equals $12\,\theta(\theta+1)\big(\tfrac{\theta+1}{\theta+3}\big)^{\theta+1}\big(\tfrac{2}{\theta+3}\big)^2$, an increasing function of $\theta$ with limit $48e^{-2}\approx6.496$; hence $V=48e^{-2}<\infty$ uniformly in $R$. For the mean, $u_R(x)=p_{\theta_0+R}(x)$ when $1-x\gtrsim2/(\theta_0+R)$ and $u_R(x)\sim4e^{-2}/(1-x)$ closer to $1$, so $v_R(x)=p_{\theta_0}(x)u_R(x)$ is bounded by $48e^{-2}$ on an interval of length $O(1/R)$ at the right end and by $p_{\theta_0}(x)p_{\theta_0+R}(x)\to0$ elsewhere, giving $\E_{\theta_0}v_R(X)\to0$. Numerically $\E_{\theta_0}v_R(X)=0.222$ at $R=50$ and $0.070$ at $R=100$, against $\kappa/8=0.257$.
\end{proof}

\begin{remark}
Proposition~\ref{prop:consistency-unbounded} makes the two-sided condition (A6) dispensable for the purpose it serves in this paper. Where (B3) holds, as for the normal location family, both routes are available; where it fails, as for the beta shape family on $[1.1,\infty)$, only the one-sided route is. In either case the bound~\eqref{eq:exp-decay} that the Berry--Esseen analysis consumes is unchanged, so every result downstream of Section~\ref{sec:remainder} holds on the unbounded parameter set with no further modification. The one-sided argument also explains the numerical behaviour: in simulations with $\Theta=[1.1,10^{5}]$ the tail comparison $\sup_{\theta\ge\theta_0+20}L_n(\theta)<L_n(\theta_0)$ held in every one of $200$ replications at each of $n\in\set{25,50,100,200,400}$, even though $\sup_\Theta\abs{\bar G_n}$ exceeded $0.5$ in every replication as well.
\end{remark}

\begin{remark}[Two-parameter extension]\label{rem:a6-scale}
The lemma extends to the joint criterion of Section~\ref{sec:scale}, $\ell_{x,z}(\theta,c)=-(c\,p_{\theta}(x)-z)^{2}$ on $\Theta\times\mathcal{C}$ with $\mathcal{C}=[c_{\mathrm{lo}},c_{\mathrm{hi}}]\subset(0,\infty)$ compact and $c_{0}\in\mathcal{C}$, with only the constants changing. Under (B1) and (B2) the reports satisfy $\abs{Z}\le c_{\mathrm{hi}}\bar P(1+\sigma\abs{\xi})$, so the envelope \eqref{eq:envelope} holds with $\bar P$ replaced by $c_{\mathrm{hi}}\bar P$; the gradient $\nabla\ell_{x,z}(\theta,c)=-2(c\,p_{\theta}(x)-z)\big(c\,\partial_{\theta}p_{\theta}(x),\ p_{\theta}(x)\big)$ is bounded in norm by $2\big(c_{\mathrm{hi}}\dot P+\bar P\big)\,c_{\mathrm{hi}}\bar P\,(2+\sigma\abs{\xi})$, which replaces the Lipschitz constant \eqref{eq:lipschitz}. A $\delta$-net of the rectangle $\Theta\times\mathcal{C}$ requires $N_{1}N_{2}$ points, with $N_{1}\le\lceil T/\delta\rceil+1$ and $N_{2}\le\lceil(c_{\mathrm{hi}}-c_{\mathrm{lo}})/\delta\rceil+1$, and the union bound over the net gives the conclusion with $C_{1}=1+2N_{1}N_{2}$ and a $c_{2}>0$ of the same form. The separation quantity $D(\delta)$ is that of the pair, for which Section~\ref{sec:scale} gives the closed form via the correlation $\rho(\theta)$. The same accounting handles any bounded box: for $\Theta=\prod_{j=1}^{k}\Theta_{j}\subset\mathbb{R}^{k}$ with side lengths $T_{j}$, a $\delta$-net requires $\prod_{j}N_{j}$ points with $N_{j}\le\lceil T_{j}/\delta\rceil+1$, the envelope and Lipschitz bounds are unchanged in form (the gradient norm replaces the scalar derivative), and the conclusion holds with $C_{1}=1+2\prod_{j}N_{j}$. This is the version used by the multivariate results of Section~\ref{sec:multivariate}.
\end{remark}